\documentclass[11pt]{llncs}
\newtheorem{clm}[theorem]{Claim}
\usepackage{mathtools}
\usepackage{amsmath}
\usepackage{amssymb}
\usepackage[table, dvipsnames]{xcolor}
\usepackage{adjustbox}
\usepackage{array, makecell}
\usepackage{framed}
\usepackage{setspace}
\usepackage{graphicx}
\usepackage{bbm}
\usepackage{amsfonts}
\usepackage{tikz}
\usetikzlibrary{decorations.pathreplacing}
\usepackage[breaklinks=true]{hyperref}
\hypersetup{colorlinks=true,%
            citebordercolor={.6 .6 .6},linkbordercolor={.6 .6 .6},%
citecolor=blue,urlcolor=black,linkcolor=blue}
\usepackage[nameinlink]{cleveref}
\Crefname{algocf}{Algorithm}{Algorithms}
\crefname{algocfline}{line}{lines}
\Crefname{invariant}{Invariant}{Invariants}
\Crefname{clm}{clm}{Claims}
\Crefname{clm}{clm}{Claims}
\Crefname{figure}{Figure}{Figures}
\Crefname{subclaim}{Subclaim}{Subclaims}

\usepackage{breakcites}
\usepackage{scalerel}[2016/12/29] % To scale z in exponent
\usepackage{multirow}
\usepackage{tikz}
\usepackage{algpseudocode}
\usepackage{boxedminipage}

\usepackage[ruled,vlined,linesnumbered,algonl]{algorithm2e}

\SetCommentSty{mycommfont}
\DontPrintSemicolon
\usepackage{newfloat}
\DeclareFloatingEnvironment[fileext=lop]{Algorithm}

\DeclareMathOperator*{\argmin}{arg\,min}

\newcommand{\Mid}{{\mathsf{Mid}}}
\newcommand{\Bal}{{\mathsf{Bal}}}
\newcommand{\X}{\mathcal{X}}
\newcommand{\veps}{\varepsilon}
\newcommand{\eps}{\varepsilon}
\newcommand{\E}{\mathbf{E}}
\newcommand{\D}{\tilde{D}}

\newcommand{\cS}{{\mathcal S}}

\newcommand{\tpi}{{\widetilde \pi}}
\newcommand{\tO}{{\widetilde O}}

\definecolor{mycolor}{rgb}{1, 0.0, 0.0}

\allowdisplaybreaks
\newcommand{\supp}{{\textsf{supp}}}
\newcommand{\loss}{{\textsf{loss}}}
\newcommand{\conloss}{{\textsf{concatloss}}}
\newcommand{\robust}{{\textsf{Robust Sort}}\xspace}
\newcommand{\LCS}{{\textsf{LCS}}}
\newcommand{\lcs}{{\textsf{lcs}}}
\newcommand{\fvs}{{\textsf{Feedback Vertex Set}}\xspace}
\newcommand{\RS}{{\texttt{ROBUST-SORT}}}

\begin{document}

\title{Robust-Sorting and Applications to Ulam-Median}
\author{Ragesh Jaiswal, Amit Kumar, Jatin Yadav}
\institute{Department of Computer Science and Engineering, IIT Delhi}

\maketitle

\begin{abstract}
Sorting is one of the most basic primitives in many algorithms and data analysis tasks. Comparison-based sorting algorithms, like quick-sort and merge-sort, are known to be optimal when the outcome of each comparison is error-free. However, many real-world sorting applications operate in scenarios where the outcome of each comparison can be noisy. 
In this work, we explore settings where a bounded number of comparisons are potentially corrupted by erroneous agents, resulting in arbitrary, adversarial outcomes.

\ \ \ \ \ We model the sorting problem as a query-limited tournament graph where edges involving erroneous nodes may yield arbitrary results. Our primary contribution is a randomized algorithm inspired by quick-sort that, in expectation, produces an ordering close to the true total order while only querying $\tilde{O}(n)$ edges. 
We achieve a distance from the target order $\pi$ within $(3 + \epsilon)|B|$, where $B$ is the set of erroneous nodes, balancing the competing objectives of minimizing both query complexity and misalignment with $\pi$. Our algorithm needs to carefully balance two aspects -- identify a pivot that partitions the vertex set evenly and ensure that this partition is ``truthful'' and yet query as few ``triangles'' in the graph $G$ as possible. Since the nodes in $B$ can potentially {\em hide} in an intricate manner, our algorithm requires several technical steps that ensure that progress is made in each recursive step. 

\ \ \ \ \ Additionally, we demonstrate significant implications for the Ulam-$k$-Median problem. This is a classical clustering problem where the metric is defined on the set of permutations on a set of $d$ elements.  Chakraborty, Das, and Krauthgamer gave a $(2-\varepsilon)$ FPT approximation algorithm for this problem, where the running time is super-linear in both $n$ and $d$. We give the first $(2-\varepsilon)$ FPT linear time approximation algorithm for this problem. Our main technical result gives a strengthening of the results in Chakraborty et al. by showing that a good 1-median solution can be obtained from a constant-size random sample of the input. We use our robust sorting framework to find a good solution from such a random sample. We feel that the notion of robust sorting should have applications in several such settings. 
\end{abstract}
\section{Introduction}
Sorting is one of the most basic primitives in many algorithms and data analysis tasks. The classical model 
of comparison-based sorting has been extensively studied, where one aims to sort 
a list of $n$ objects using pairwise comparisons. It is well-known that, in 
this model, sorting requires at least $ n \log n $ comparisons in the worst case. 
Popular algorithms such as merge sort, quick-sort, and heap-sort achieve 
this bound with $ O(n \log n) $ comparisons.

However, in practical scenarios, sorting is often applied to very large datasets 
where errors or imperfections in the comparisons are unavoidable. In real-world 
applications involving noisy data or large-scale distributed systems, comparisons 
may occasionally be faulty due to hardware imperfections, data corruption, or other 
noisy sources. Thus, it is crucial to extend classical sorting algorithms to handle 
such imperfections effectively.
An approach for dealing with such noisy errors in sorting was initiated by Feige, Raghavan, Peleg, and Upfal \cite{FRPU94}. In their model, 
each comparison's outcome is flipped independently with some known probability $p$. Assuming $p$ is a constant, repeatedly querying a pair $\theta(\log n)$ times would ensure that the outcome is correct with high probability. However, this would lead to $O(n \log^2 n)$ number of comparisons.~\cite{FRPU94} showed that one could instead obtain an algorithm that outputs the sorted order with high probability and performs $O(n \log n)$ comparisons. There has been much recent work~\cite{wgw22,gx23} on obtaining tight dependence on the parameter $p$.

Our work is rooted in adversarial settings where each key is controlled by an agent, and some of the agents may be erroneous. Thus, any comparison involving one of the erroneous agents could result in an arbitrary, though deterministic, outcome. Such models are of practical relevance. For example, in distributed computing environments, 
each data element may be processed by different nodes, and some nodes could behave 
erratically due to hardware failures or software bugs. Similarly, in data integration 
tasks, comparisons may be unreliable due to discrepancies between data sources with 
inconsistent quality. In these scenarios, it is natural to seek sorting algorithms 
that are robust to a bounded number of adversarial errors.

To formalize this setting, we model the outcome of all pair-wise comparisons between $n$ keys as a tournament graph $G$ (recall that a tournament is a directed graph that has a directed arc between each pair of vertices) and assume that there is a total order $\pi$ on the vertices of $G$. If there were no erroneous edges, then the edges of $G$ would be consistent with $\pi$ (i.e., $G$ has a directed arc $(u,v)$ iff $u$ appears before $v$ in $\pi$). However, there are some ``bad'' nodes $B$ in this graph: querying any edge involving such a bad node as an end-point can lead to an (adversarially) arbitrary outcome. But the outcome of querying any other edge is consistent with $\pi$.  

The algorithm queries some of the edges of $G$, and then outputs an ordering on all the vertices of $G$. 
The goal is to output an ordering with minimum distance (measured in terms of the length of the longest common subsequence) with respect to $\pi$. 

% If we were allowed to ask $\Omega(n^2)$ queries, there is a simple algorithm that outputs an ordering with a distance at most $3|B|$ from $\pi$. 
Observe that there are two competing objectives here: query as few edges of $G$ as possible and minimize the mismatch with the hidden input sequence. If we query all the $O(n^2)$ edges of $G$, it is not difficult to show that a 2-approximation algorithm for the feedback vertex set on tournaments~\cite{fvst} can be used to obtain an ordering that has distance at most $3|B|$ from $\pi$.  In this work, we address the following question: 
\begin{quote}
    \emph{Can we get an efficient algorithm that queries $\tilde{O}(n)$ edges in $G$, and outputs an ordering that has distance at most $O(|B|)$ from $\pi$? }
\end{quote}

We answer this question in the affirmative and give a randomized algorithm that comes within distance  $(3+\epsilon)|B|$ of $\pi$. Furthermore, it runs in $\tilde{O}(n)$ time.
Our algorithm is based on careful identification of good pivots which are likely to be outside the set $B$. Since the algorithm does not know the set $B$, we can only use indirect approaches for this. Our algorithm relies on finding directed triangles in $G$ -- we know that there must be at least one vertex of $B$ in such a triangle. However, each failed attempt to identify such a triangle increases the query complexity. Thus, the algorithm needs to carefully balance these two aspects -- identify a pivot that partitions the vertex set evenly and ensure that this partition is ``truthful'' and yet query as few triangles in the graph $G$ as possible. Direct implementations of these ideas do not work for several subtle reasons: (i) a pivot belonging to the set $B$ may be able to {\em hide} itself by participating in few triangles, and yet, may create large number of misalignments if used for partitioning the set of elements, (ii) elements in $V \setminus B$ will result in truthful partitions of $V$, but may be involved in lot of triangles (involving elements in $B$), and hence, the algorithm may not use them to act as pivots. The solution lies in a technically more involved strategy: first, ensure by random sampling that there are not too many triangles in $G$. Subsequently, we show that an element that results in an almost balanced partition and is not involved in too many triangles can be used as a pivot. Proving this fact lies at the heart of our technical contribution. \\

\noindent {\bf{Implications for Ulam-$k$-Median Problem.}} Our result on robust sorting has interesting implications for another well-studied problem, namely the Ulam-$k$-Median problem. Given a positive integer $d$, let $[d]$ denote the set $\{1, \ldots, d\}$ and $\Pi^d$ the set of all permutations over $[d]$. We can define a metric $\Delta$ over $\Pi^d$ as follows: given $\sigma_1, \sigma_2$, $\Delta(\sigma_1, \sigma_2) \coloneqq d-|\lcs(\sigma_1, \sigma_2)|$, where $\lcs(\sigma_1, \sigma_2)$ denotes the length of the longest common subsequence between $\sigma_1$ and $\sigma_2$. The metric $(\Pi^d, \Delta)$ is also popularly known as the Ulam metric.
An instance of this problem is specified by a subset $S \subseteq \Pi^d$ and an integer $k$. The goal is to find a subset $F$ of $k$ permutations in $\Pi^d$ (which may not lie in $S$) such that the objective value
$$ \sum_{\sigma \in S} \min_{\pi \in F} \Delta(\pi, \sigma)$$
is minimized. In other words, it is the $k$-median problem where the metric is given by $\Delta$ on $\Pi^d$. When $k=1$, there is a simple 2-approximation algorithm (which works for any metric): output the best point in the input set $S$. 

Breaking the approximation guarantee of $2$ for specific metrics, such as the Ulam metric, was open for a long time. Recently, a sequence of works~\cite{cdk21,cdk23} breaks this 2-factor barrier for the Ulam $k$-median problem. The algorithm is a fixed parameter tractable (FPT) algorithm (with respect to the parameter $k$) that gives an approximation guarantee of $(2 - \delta)$ for a very small but fixed constant $\delta > 0$. The running time of this algorithm is $(k \log{nd})^{O(k)} nd^3$, which can be written as $f(k)\cdot poly(nd)$ and hence is FPT time.
Note that the running time is super-linear in $n$ and cubic in $d$.
This contrasts with several FPT approximation algorithms for the $k$-median/means problems in the literature~\cite{kss,jks,bjk18,bgjk20,gj20}  that have linear running time in the input size. For example, the running time of the FPT $(1+\veps)$-approximation algorithms for the Euclidean $k$-median/means problems in \cite{kss,jks,bjk18,bgjk20}  is linear in $nd$, where $d$ is the dimension. Similarly, the running time for the FPT constant factor approximation algorithms~\cite{gj20} in the metric setting is linear in $n$. Thus, we ask: 
\begin{quote}
{\it Is it possible to break the barrier of 2-approximation for the Ulam $k$-median problem using an FPT algorithm with a \textbf{linear} running time in the input size, i.e., $nd$ ?}
\end{quote}

We show that the answer is affirmative using two crucial ideas, one of which relies on the robust sorting problem.
We build on the ideas of \cite{cdk23} for designing a $(2-\delta)$-approximation for the 1-median problem. They show that for any input set $S$, there exist five permutations $\sigma_1, ..., \sigma_5$ in $S$ from which we can find a permutation $\tilde{\sigma}$ that has a small distance with respect to the optimal $\sigma^*$. To obtain linear in $n$ dependence on the running time, we show that a constant size randomly sampled subset of input permutations contains these five permutations with high probability. To obtain linear dependence on $d$, given guesses $\sigma_1, ..., \sigma_5$ for the five permutations, we can assign each pair of symbols $(a,b)$ a direction based on the {\em majority} vote among these permutations. Now, it turns out that a good solution to the corresponding robust sorting instance is close to the optimal solution $\sigma^*$.

\section{Preliminaries and Problem Definition}
We discuss two problems -- {\em Robust Sorting and Ulam median}. We start with the robust sorting.\\

\noindent {\bf Robust Sorting:} 
We are given a set of $n$ elements $V$. There is an (unknown) total order $\pi$ over $V$. However, we have imperfect access to the total order $\pi$. We are given an implicit directed graph $G = (V, E)$, where we have a directed edge between every pair of elements in $V$ (such graphs are often denoted {\em tournaments}). In the ideal (zero error) scenario, the edge set $E$ would correspond to $\pi$, i.e., for every distinct pair $u, v \in V$, we have the directed arc $(u,v)$ iff $u$ comes before $v$ in $\pi$. However, we allow the graph $G$ to represent $\pi$ in an imperfect manner as formalized below: 

\begin{definition}[Imperfect representation]
We say that a tournament $G=(V, E)$ is $B$-imperfect with respect to a total order $\pi$ on $V$, where $B$ is a subset of $V$, if for every pair of distinct $u,v \in V \setminus B$, $G$ has the arc $(u,v)$  iff $u$ comes before $v$ in $\pi$. For an integer $b$, we say that $G$ is $b$-imperfect w.r.t. $\pi$ if there is a subset $B$ of $V$, $|B| \leq b$ such that $G$ is $B$-imperfect with respect to $\pi$. 
\end{definition}
In other words, if $G$ is $B$-imperfect w.r.t. $\pi$, then the arcs with both end-points outside $B$ represent $\pi$ correctly, but we cannot give any guarantee for arcs incident with vertices in $B$. We shall often ignore the reference to $\pi$ if it is clear from the context. For edge $(u, v) \in E$, we will sometimes use the notation $u < v$ or $v > u$. Note that the relations $<, >$ are not transitive, except when $|B| = 0$. Similarly, we will sometimes say that $u$ is {\it lesser than} (resp. {\it greater than}) $v$ if $(u, v) \in E$ (resp. $(v, u) \in E$).

The \robust problem is as follows: given a set of points $V$ with an implicit total order $\pi$, and a query access to the edges in a tournament  $G$, output an ordering $\tpi$ on $V$ which maximizes $\lcs(\pi, \tpi)$ while using as few queries to $G$ as possible.  Here $\LCS(\pi, \tpi)$ denotes the  longest common subsequence between $\pi$ and $\tpi,$ and $\lcs(\pi, \tpi)$ denotes the length of $\LCS(\pi, \tpi)$. Observe that when $G$ is $0$-imperfect w.r.t. $\pi$, one can obtain the total order $\pi$ with $O(n \log n)$ queries to $G$.

The \robust problem can be reduced to the well-studied \fvs problem on tournaments (FVST).
A feedback vertex set in a directed graph is a subset of vertices whose removal makes the graph acyclic. The Feedback Vertex Set in Tournaments (FVST) problem is to find a feedback vertex set of the smallest size in a given Tournament graph.
%\begin{definition}[\fvs]
%    \textcolor{RoyalBlue}{Given a directed graph $G = (V, E)$, a subset of vertices is called a feedback vertex set if its removal makes the graph acyclic. In the \fvs problem, the goal is to find the smallest feedback vertex set of the input graph $G$.}
%\end{definition}
%Observe that if $\tpi$ is an ordering on $V$ with $\LCS(\pi, \tpi)$ being given by the restriction of $\pi$ (or $\tpi$) on a subset $T$, then $G[T]$ is a DAG; and the converse holds as well. 
If we are allowed to query all the edges in $G$, then an $\alpha$-approximation algorithm for FVST can be utilized to obtain an $(\alpha+1)$-approximation algorithm for \robust as follows: {\em find an $\alpha$-approximate feedback vertex set $B'$ and then find a topological ordering $\pi_1$ on $V \setminus B'$. Output the concatenation $\tpi$ of any arbitrary ordering on $B'$ with $\pi_1$}. 
Clearly, $\LCS(\pi, \tpi) \geq |\pi_1| - |B| = |V| - |B'| - |B| \geq |V| - (\alpha+1)|B|.$ Hence, using exponential time, we can find an optimal feedback vertex set and therefore a $2$-approximate solution to \robust. However, since feedback vertex set on tournaments is NP-hard to approximate better than a factor of $2$, and a polynomial time 2-approximation algorithm is known~\cite{fvst}, $3$-approximation is the best we can get from such a direct reduction to FVST. Indeed, consider a simple example where $V = \{v_1, v_2, v_3, v_4\}$, $\pi = (v_2, v_1, v_3, v_4), B = \{v_2\}$. Consider a $B$-imperfect tournament $G$, such that for all unordered pairs $(v_i, v_j)$ with $i < j$, there is an arc from $v_i$ to $v_j$ in $G$, except the direction of the arc is reversed for $(v_2, v_3)$ and $(v_2, v_4)$. In this example, a $2$-approximation algorithm might return $\{v_3, v_4\}$ as the feedback vertex set. Hence, $\pi_1 = (v_1, v_2)$ is the only topological ordering on the remaining vertices, and we might return $\tpi = (v_4, v_3, v_1, v_2)$. Here, $\LCS(\pi, \tpi) = 1 = |V| - 3|B|.$
% \textcolor{BrickRed}{Note that the above $3$-approximation algorithm for \robust makes $\Omega(n^2)$ edge queries since the FVST problem requires the complete tournament description as input.}

Although the algorithm of  Lokshtanov et al.~\cite{fvst} is also inspired by quick-sort, it queries $\Omega(n^2)$ edges and runs in $O(n^{12})$ time. In our setting, we are constrained by near-linear running time and, hence, can only check a small subset of triangles for consistency. However, this causes several other issues: a bad pivot can masquerade as a good one, and a good pivot may not get a chance to partition the elements into almost equal halves. The fact that we are on a very tight budget in terms of the number of queries and errors makes the algorithm and analysis quite subtle. \\

\noindent {\bf Ulam median: } The Ulam $ k$-median problem is simply the $k$-median problem defined over the Ulam metric $(\Pi^d, \Delta)$ -- {\it given a set $S \subset \Pi^d$ of $n$ elements and a positive integer $k$, find a set $C$ of $k$ elements (called centers) such that the objective function $Obj(S, C) \coloneqq \sum_{s \in S} \min_{c \in C} \Delta(s, c)$ is minimized.}
Here, $\Pi^d$ is the set of all permutations of $[d] \coloneqq \{1, 2, 3, ..., d\}$, which can also be seen as all $d$-length strings with distinct symbols from set $\{1, 2, ..., d\}$.
The distance function $\Delta$ is defined as $\Delta(x, y) \coloneq d - \lcs(x, y)$, where $\lcs(x, y)$ denotes the length of the {\em longest common subsequence} of permutations $x, y \in \Pi^d$. There is a trivial 2-approximation algorithm and it has been a long open challenge to obtain a better approximation. Breaking this barrier of 2 was recently achieved by Chakraborty et al.~\cite{cdk23} who gave a parameterized algorithm (with parameter $k$) with a running time of $(k \log{nd})^{O(k)} nd^3$ and approximation guarantee of $(2-\delta)$ for a small constant $\delta$. Our goal was to achieve the same using a parameterized algorithm with running time $\tilde{O}(f(k) \cdot nd)$.

\section{Our Results}
Our first result gives an algorithm for \robust that queries $\tO(n)$ (here $\tO$ hides poly-logarithmic factors) edges in $G$ and achieves nearly the same guarantees as that obtained by an efficient algorithm querying all the edges in $G$ followed by a reduction to FVST.
\begin{theorem} \label{thm:rsort}
    Consider an instance of $\robust$ given by a tournament graph $G=(V,E)$, where $|V| = n$,  and a parameter 
   $\varepsilon > 0$. Suppose $G$ is $b$-imperfect w.r.t. an ordering $\pi$ on $V$. Then, there is an efficient algorithm that queries $O\left(\dfrac{n \log^3 n}{\varepsilon^2}\right)$ edges in $G$ and outputs a sequence $\tpi$ such that expectation of $\lcs(\pi, \tpi)$ is at least  $n-(3+\varepsilon)b.$ Further, the algorithm does not require  knowledge of the quantity $b$ and has running time $O\left(\dfrac{n \log^3 n}{\varepsilon^2}\right)$ (assuming each query takes constant time). 
\end{theorem}

It is worth emphasizing that the parameter $b$ may not be constant. In fact, much of the technical difficulty lies in handling cases when $b$ may be sub-linear. 
Following is our main result for the Ulam $k$-median problem.
\begin{theorem}
There is a randomized algorithm for the Ulam $k$-median problem that runs in time $\tilde{O}((2k)^k \cdot nd)$ and returns a center set $C$ with $Obj(S, C) \leq (2-\delta) \cdot OPT$ with probability at least $0.99$, where $\delta$ is a small but fixed constant.
\end{theorem}

\subsection{Our Techniques}
We now give an overview of our techniques for the robust sorting and the Ulam-$k$-median problems. 

\subsubsection{Robust Sorting:}
Consider an instance of \RS\  given by a graph $G$ on a vertex set $V$ of size $n$. Assume $G$ is $B$-imperfect for some $B \subseteq V$, $|B|=b$. 
We shall call the elements of $B$     ``bad'' elements and the rest ``good'' elements. Observe that every directed triangle in $G$ must contain at least one bad element. One potential idea is to keep finding and removing directed triangles in $G$, until $G$ is acyclic (recall that a tournament is acyclis if and only if it has no triangles). Suppose we remove $t$ triangles. As each removed triangle has at least one bad element, $t \leq b$ and the number of remaining bad elements is at most $b - t$. Hence, the number of remaining good elements is at most $n - b - 2t - (b - t) = n - 2b - t \geq n - 3b$. Thus, we produce an ordering $\tpi$ with $\lcs(\pi, \tpi) \geq n - 3 b$. However, this approach uses a quadratic number of queries on $G$. On the other hand, there are many sorting algorithms that use $O(n \log n)$ queries in the classical setting, i.e., the 0-imperfect setting. Direct generalizations of such sorting algorithms fail to get the required guarantees on the $\lcs$ between the output $\tpi$ and $\pi$. 

For example, consider Merge Sort. 
Here, bad elements can result in the merge procedure to output permutations with a large distance with respect to $\pi$. Indeed, suppose we partition the input set $V$ into equal sized $V_1, V_2$ and recursively get good orderings $\tpi_1$ and $\tpi_2$ on them, respectively. Further, assume that all the good elements in $V_1$ appear before those in $V_2$. However, it is possible that the first element in $\tpi_1$ happens to be a bad element $x$, and $x$ turns out to be larger than all the elements in $V_2$. In such a setting, the merge procedure would place all the elements in $V_2$ before $x$, which is clearly an undesirable outcome.

We now show that using randomized quick sort directly would also lead to an undesirable outcome. In randomized quick sort, one chooses a pivot at random and recursively solves the problem on the elements smaller than the pivot and those larger than the pivot, placing the pivot between the two recursive outputs. For the feedback arc set in tournaments (FAST) problem, where the goal is to minimize the number of inverted pairs, this algorithm was shown to return a $3$-approximation~\cite{acn08} in expectation. However, this simple algorithm does not work for our problem. Indeed, let $b \ll n$ and consider a random input where a subset $B$ of size $b$ is chosen at random as the set of bad elements, and a random permutation is selected among the set of good elements. Edges where both endpoints are good respect the permutation, and each edge incident on a bad element is oriented randomly. Now, the random quick-sort algorithm chooses a bad pivot $x$ with probability $b/n$ -- assume this event happens. Let $V_1, V_2$ be the elements less than and greater than $x$, respectively (with respect to the graph $G$). Since $x$ is a bad pivot, our assumption implies that $V_1, V_2$ form a roughly random partition of $V$ into equal-sized subsets. Thus, if $X \coloneqq V \setminus B$ denote the set of good elements and $X_1, X_2$ be the left and the right half of $X$ respectively, then roughly $|X_1|/2$ elements of $X_1$ will end up in $V_2$ and similarly for $X_2$. Note that $|X_1| = |X_2| = \frac{n-b}{2}$. Hence, roughly $(n-b)/4$ elements will be {\em wrongly} placed by the pivot $x$. Let $f(n,b)$ denote the distance between the output produced by this algorithm on an instance of size $n$ containing $b$ bad elements and the ordering $\pi$. Then, we have shown that $f(n,b)$ is at least $\Omega(n-b)$ with a high probability if we choose a bad pivot. Thus, we get the approximate recurrence (note that both $V_1, V_2$ will have roughly $b/2$ bad elements): 
$$f(n,b) \approx \frac{b}{n} \cdot \Omega(n-b) + 2 f(n/2,b/2)$$
It is easy to verify that this results in $f(n,b) = \Omega(b \log n)$, whereas we desire an output for which this quantity is $O(b)$.

The above analysis indicates that we cannot afford to have ``arbitrarily'' bad elements as pivots. The following idea takes care of examples as above: when we choose a pivot $p$, check (logarithmic number of times)  if $p$ forms a triangle with randomly sampled pair $(x,y)$. If a triangle is found, we can remove $p, x, y$ (and hence, at least one bad element gets removed) and try a new pivot; otherwise, it is guaranteed that $p$ would be involved in very few triangles (note that in the example above, a triangle would be found with constant probability if the pivot is a bad one).  However, this simple idea also does not work. The reason is as follows: (i) randomly chosen good elements, which would ideally act as good pivots, are involved in a lot of triangles and, hence, would not be selected as pivots, and (ii) there could be bad elements, which are not involved in too many triangles, and hence, would sneak in as a pivot. It may seem that the latter scenario is not undesirable  -- if a bad element participates in a few triangles, then it perhaps acts like a good element and can be used to partition the input set. Unfortunately the quantity $f(n,b)$ as defined above turns out to be large for this algorithm.  This happens for the following subtle reason: say there are $b$ bad elements, and suppose each of the good elements participates in many triangles. Then, with probability $(1-b/n)$ (which can be considered to be close to 1), the algorithm picks a good element as a random pivot and finds a triangle containing it. This would reduce the problem size by only three elements, i.e., the recursive problem has almost the same size. On the other hand, with probability $b/n$, which may be small, the algorithm partitions using a bad element as a pivot. As outlined above, when we pick a bad element as a pivot, the resulting partition may have many good elements on the wrong side of the partition. Thus, in the overall calculation, the large misalignment created due to these low probability events overwhelms the expected value of $f(n,b)$. In other words, one obtains a recurrence of the form: 
$$f(n,b) \approx \frac{b}{n} \cdot m_b + f(n-3,b-1)$$
where $f(n-3,b-1)$ refers to the sub-problem obtained when a triangle is found, and $m_b$ refers to the misalignment caused by a typical bad element. Since a bad element participates in few triangles, it is possible that $m_b \ll \Omega(n-b)$ (comparing with the previous recurrence above), but still, this can be high enough to lead to a recurrence where $f(n,b)$ is not $O(b)$. The issue arises because the problem size does not shrink sufficiently to balance the expected misalignment. One way to handle this is to consider separately the case where there are too many triangles in the tournament. So, before picking a potential pivot and testing its goodness, in a {\em pre-pivoting step}, we check randomly chosen triples for triangles and remove them in case they are found. This ensures that at the time of pivot selection, there is a reasonable chance the problem size shrinks without too much increase in the expected misalignment. Our algorithm and analysis basically work by balancing these quantities in a carefully devised inductive argument. One of the key technical insights is the following: given two orderings $\sigma_1$ and $\sigma_2$ of a partition $V_1$ and $V_2$ of $V$ respectively, we define the notion of {\em concatenation loss} -- this captures the extra misalignment (with respect to the implicit ordering) created by concatenating $\sigma_1$ and $\sigma_2$. Our key technical result shows that if such a partitioning is created by a pivot involved in a few triangles, then the corresponding concatenation loss of the orderings on $V_1$ and $V_2$ output by the algorithm is small.

\subsubsection{Ulam-$k$-Median}
There is a trivial and well-known 2-approximation algorithm for the Ulam 1-median problem -- {\it output the best permutation from the input.}  To break the barrier of 2, one must consider a stronger version of the triangle inequality that holds specifically for the Ulam metric. Let $\sigma_1, ..., \sigma_n$ be the permutations in the input and let $\sigma^*$ denote the optimal 1-median. Let $I_{\sigma_i}$ denote the subset of symbols that are not in $\LCS(\sigma_i, \sigma^*)$, i.e., the misaligned symbols in $\sigma_i$ and $\sigma^*$. So, $\Delta(\sigma_i, \sigma^*) = |I_{\sigma_i}|$. The following (stronger version) of the triangle inequality holds for the Ulam metric: $\Delta(\sigma_i, \sigma_j) \leq |I_{\sigma_i}| + |I_{\sigma_j}| - |I_{\sigma_i} \cap I_{\sigma_j}| = \Delta(\sigma_i, \sigma^*) + \Delta(\sigma_j, \sigma^*) - |I_{\sigma_i} \cap I_{\sigma_j}|$. Chakraborty et al.~\cite{cdk23} exploit this inequality to break the 2-approximation barrier. 
Even though the technical details are intricate, at a very high level, the key idea in \cite{cdk23} is to show that either one of the input permutations is a good center or there are five permutations $\sigma_1, ..., \sigma_5$ such that $\forall i, j \in \{1,2,3,4,5\}$ with $i \neq j$, $|I_{\sigma_i} \cap I_{\sigma_j}|$ is small, i.e., the number of {\em common} misaligned symbols is small.

We strengthen this result as follows -- we show that either there are a significant number of permutations that will be good centers, or there are a significant number of permutations with a small number of pair-wise common misaligned symbols.
This allows us to argue that an $\eta$-sized (for some {\bf constant} $\eta$) random subset of permutations is sufficient to find a good center, so we do not need to consider $\binom{n}{5}$ possibilities for finding a good center. Chakraborty et al.~\cite{cdk23} gave a similar sampling lemma. However, they required a random sample of size $O(\log n)$. Using this result would lead to an additional $(\log{n})^{O(k)}$ factor in the running time for the $k$-median problem. 
One of our key contributions is showing that a constant-sized sample suffices. 
%\footnote{\rageshedit{Chakraborty et al.~\cite{cdk23} gave a similar sampling lemma where it was shown that a good centre could be found by trying all combinations of five strings from a uniform sample of size $O(\log{n})$. However, the logarithmic dependence on the sample size adds a factor of $O((\log{n})^k)$ in the running time for the Ulam-$k$-median problem. Avoiding the $O(\log{n})^k$ factor was one of our objectives, and hence showing the sampling lemma with a {\bf constant} sample size was important.}}
The number of common misaligned symbols in the five permutations being small implies that for most pairs $(a, b)$ of symbols, their relative order in $\sigma^*$ matches that in at least 3 out of 5 permutations $\sigma_1, ..., \sigma_5$. This is used to find a permutation with a good agreement with $\sigma^*$ and hence is a good center. \cite{cdk23} uses an $O(d^3)$ procedure to find such a center, whereas we improve this to $\tilde{O}(d)$ by using our robust sorting algorithm. For the Ulam $k$-median problem, we use our sampling-based algorithm, {\tt ULAM1}, within the $D$-sampling framework of \cite{jks} to obtain an $\tilde{O}(f(k) \cdot nd)$-time algorithm. Here is the summary of the key ideas: Let $S_1, ..., S_k$ be the dataset partition that denotes the optimal $k$ clustering, and let $\sigma^*_1, ..., \sigma^*_k$ denote the optimal centers, respectively. Let us try to use {\tt ULAM1} to find good centers for each of $S_1, ..., S_k$. We would need $\eta$ uniformly sampled points each from $S_1, ..., S_k$. 
The issue is that the optimal clustering $S_1, ..., S_k$ is not known.
If the clusters were balanced, i.e., $|S_1| \approx |S_2| \approx ... \approx |S_k|$, then uniformly sampling $O(\eta k)$ points from $S$ and then considering all possible partitions of these points would give the required uniform samples $X_1, ..., X_k$ from each of the optimal clusters. We can then use {\tt ULAM1($X_i$)} to find good center candidates for $S_i$. 
In the general case, where the clusters may not be balanced, we use the $D$-sampling technique\footnote{$D$-sampling is sampling proportional to the distance of an element from the closest previously chosen center.} to boost the probability of sampling from small-sized optimal clusters, which may get ignored when sampling uniformly at random from $S$. 
%We discuss the Ulam problem in the Appendix, the high-level outline in Section~\ref{sec:ulamoutline} and detailed proofs in Section~\ref{sec:ulamdetails}.

\subsection{Related Work}
We have already seen the connection between robust sorting and the Feedback Vertex Set in Tournaments (FVST) problem~\cite{fvst,mnich16}. Another problem related to the FVST problem is the Feedback Arc Set in Tournaments (FAST) problem~\cite{acn08,ms07}, where the goal is to find an ordering of the nodes of a given tournament such that the number of edges going backward (an edge is said to go backward if it is directed from a node that comes later to a node that comes earlier as per the ordering) is minimized. This is the restricted variant of the maximum acyclic subgraph problem \cite{fk99}, where the goal is to find the maximum subset of edges that induces a directed acyclic graph (DAG) in a directed graph. The FAST problem may be seen as robust sorting under adversarial corruption of edges rather than adversarial corruption of nodes, as in our formulation. The FVST and FAST problems are known to the $\mathsf{NP}$-hard. A 2-approximation, which is tight under UGC, is known~\cite{fvst} for the FVST problem, and a PTAS is known~\cite{ms07} for the FAST problem. 

Several works have been done on sorting in the presence of a noisy comparison operator, also called noisy sorting. Feige et al.~\cite{FRPU94} consider a noise model in which the comparison operator gives the correct answer independently with probability at least $(1/2 + \gamma)$ each time a query is made on an element pair. This can be regarded as {\em noisy sorting with resampling} since we can get the correct answer for a pair by repeatedly querying the operator on the same pair. So, each time a comparison needs to be made for a pair, by repeatedly querying $O(\log{n})$ times, one can obtain a $O(n \log^2{n})$ algorithm. 
A better algorithm with $O(n \log{n})$ queries can be obtained~\cite{FRPU94,kk07}.
In more recent works~\cite{wgw22,gx23}, a deeper investigation was made into the constant, which is dependent on the bias $\gamma$, hidden in the $O(n \log{n})$ sorting algorithm of \cite{FRPU94} for noisy sorting {\bf with} resampling.
A more interesting noise model was considered by Braverman and Mossel~\cite{bm08}, where the ordering algorithm cannot repeat a comparison query.\footnote{This can also be modeled by the constraint that the errors are independent but {\em persistent}, i.e., if a comparison is repeated, then you get the same answer.}
This is called the {\em noisy sorting \textbf{without} resampling (NSWR)} problem. The NSWR problem can also be seen as a stochastic version of the Feedback Arc Set on Tournament (FAST) problem -- the tournament is generated using the noisy comparator (with respect to some total order $\pi$), and the goal is to find an ordering of the vertices such that the number of edges going backward is minimized. \cite{bm08} gave an algorithm that runs in time $n^{O((\beta+1)\gamma^{-4})}$ and outputs an optimal ordering with probability at least $(1-n^{-\beta})$. Another objective function in this setting is to minimize the maximum dislocation and total dislocation of elements, where the dislocation of an element is the difference between its ranks in $\pi$ and the output ordering. 
Optimal bounds for this have been achieved in a recent sequence of works~\cite{gllp20,gllp19}.

The key references~\cite{cdk21,cdk23} for the Ulam median problem have already been discussed. The detailed discussions on the Ulam median problem are given in ~\Cref{sec:ulamdetails}

\section{Algorithm for \robust}\label{sec:robustsort}
In this section, we present our algorithm for \robust (\Cref{algo:robust}). The algorithm discards a subset $V'$ of $V$ and returns an ordering $\pi'$ on the remaining elements $V \setminus V'$. We can obtain an ordering $\tpi$ on $V$ by appending an arbitrary ordering on $V'$ to $\pi'$. Since the size of $V'$ shall turn out to be small, $\lcs(\pi, \pi')$ will be close to $\lcs(\pi, \tpi)$. 

\begin{algorithm}[htbp]
    \caption{\RS($S$)}
    \label{algo:robust}
    {\bf Input:} A subset $S$ of the original set $T$ of elements\\
    {\bf Output:} An ordering $\pi'$ on a  subset of $S$\\
    {\bf Global: } $N = |T|, \epsilon$ \\
    
    $k \gets \left(\frac{10000 \log N}{\epsilon} \right)^2$ \tcc*{A parameter deciding the amount of testing to be done} \label{l:choosek}
    
    \While{$|S| > 0$ \label{l:while}}{ \label{l:repeat}
    
    \For{$i = 1, 2, \ldots 72  k \cdot \log N$}{ \label{l:triangle_test_begin}
        Sample three elements $x, y, z$ independently and uniformly at random from $S$\\
        \If{$x, y, z$ form a triangle}{
            $S \gets S \setminus \{x, y, z\}$ \label{l:discardtriangle}\\
            Go to line \ref{l:repeat} \label{l:triangle_test_end}
        }
    }
    
    \For{$r = 1, 2,  \ldots 36 \cdot \log N$}{ \label{l:no_triangle_found}
        Sample an element $p \in S$ uniformly at random \label{l:samplepivot}
        \tcc*{Choose a random pivot}
        $L_1, R_1 \gets \phi, \phi$ \label{l:L_1_defined}\\
        $k' \gets 10^5 \cdot \log N$ \\
        \For{$j = 1, 2, \ldots k'$ \label{l:forloop}}{ \label{l:partition_test_begin}
            Sample an element $s \in S \setminus \{p\}$ uniformly at random \\
            \If{$s < p$}{
                $L_1 \gets L_1 \cup \{s\}$
            } \Else{
                $R_1 \gets R_1 \cup \{s\}$ \label{l:partitionend}
            }
        }
        \If{$\min(|L_1|, |R_1|) \leq \frac{k'}{5} + \frac{k'}{40}$ \label{l:test}}{
            % \If{$counter > 36 \log N$\tcc*{Unable to find a balanced partition after many tries}}{
            %         {\bf return} $\{\}$ \label{l: empty_return}
            %         \tcc*{Return an empty sequence}
            %     }
                Go to the next iteration of the line \ref{l:no_triangle_found} \label{l:partition_test_end} \tcc*{Imbalanced partition}
            }
            % $counter \gets 0$  \\
        \For{$j=1, 2, \ldots k$ \label{l:ktrials}}{ \label{l:balanced_pivot}
            Sample two elements $x, y$ independently and uniformly at random from $S \setminus \{p\}$\\
            \If{$x, y, p$ form a triangle}{
                $S \gets S \setminus \{x, y, p\}$ \label{l:checktriangle}\\
                Go to line \ref{l:repeat}
            }
        } \label{l:triangle_check_end}
            
        $L, R \gets \phi, \phi$ \label{l:endpart1} \\
        \For{$s \in S$}{
            \If{$s < p$}{
                $L \gets L \cup \{s\}$
            } \Else{
                $R \gets R \cup \{s\}$ \label{l:endpart2}
            }
        }
        {\bf return} $\pi'$= (\RS($L$), $p$, \RS($R$)) \tcc*{Recursively sort $L$ and $R$} \label{l:recursive_calls}
        }
        {\bf return} $\{\}$ \label{l:no_balanced_pivot} \tcc*{Too many bad elements, return an empty sequence}
    }
    {\bf return} $\{\}$ \tcc*{No more elements left, return an empty sequence}
    \end{algorithm}
    
\Cref{algo:robust}  is based on a divide-and-conquer strategy similar to quick-sort. Consider a recursive sub-problem given by a subset $S$ of $V$. We choose a parameter $k = O\left(\frac{\log^2 N}{\epsilon^2}\right)$ (line~\ref{l:choosek}), where $N = |V|$, and then proceeds in four steps:
\begin{enumerate}

\item {\bf Testing random triplets for triangles}: In this step, we randomly sample $O( k\log N)$ triplets of elements  from $S$ uniformly at random (line~\ref{l:triangle_test_begin}). For each such triplet $(x,y,z)$, we check if it forms a triangle, i.e., if $G$ contains the arcs $(x,y),(y,z)$ and $(z,x)$. If so, we discard (elements in) the triangle (line~\ref{l:discardtriangle}) and go back to the beginning of the procedure. Note that checking whether a triplet is a triangle requires three queries to $G$.  
    
\item {\bf Finding a balanced pivot}: In this step, it tries to find a good pivot -- this pivot finding step is repeated $O(\log N)$ times (line~\ref{l:no_triangle_found}) to ensure that one of these succeeds with high probability). We first select a randomly chosen element $p$ of $S$ as the pivot (line~\ref{l:samplepivot}). Then we check whether it is a {\em balanced} pivot as follows: we sample (with replacement) $k' \coloneqq O(\log N)$ elements from $S$ and partition these $k'$ elements with respect to $p$ (lines~\ref{l:partition_test_begin}--\ref{l:partitionend}). Let $L_1$ and $R_1$ denote the partitioning of these $k'$ elements with respect to $p$. In line~\ref{l:test}, we check if both these sets are of size at least $k'/5 + k'/40$. If not, we repeat the process of finding a pivot (line~\ref{l:partition_test_end}). Otherwise, we continue to the next step.  Note that if this pivot finding iteration fails for $36 \log N$ trials, then we discard all elements of $S$ (line~\ref{l:no_balanced_pivot}).

\item {\bf Testing for triangles involving the pivot}: In this step, we test if the balanced pivot $p$ chosen in the previous step forms a triangle with randomly chosen pairs of elements from $S$. More formally, we repeat the following process $k$ times (line~\ref{l:ktrials}): sample two elements $x$ and $y$ (with replacement) from $S \setminus \{p\}$. If $(x,y,p)$ forms a triangle, we discard these three points from $S$ (line~\ref{l:checktriangle}) and go back to the beginning of the procedure (line~\ref{l:repeat}). 

\item {\bf Recursively solving the subproblems}: Assuming the pivot $p$ found in the second step above does not yield a triangle in the third step above, we partition the entire set $S$ with respect to $p$ into two sets $L$ and $R$ respectively (lines~\ref{l:endpart1}--\ref{l:endpart2}). We recursively call the algorithm on $L$ and $R$ and output the concatenation of the orderings returned by the recursive calls (line~\ref{l:recursive_calls}). 
\end{enumerate}

\section{Analysis}
In this section, we analyse~\Cref{algo:robust}. Let the input instance be given by a tournament $G=(V, E)$, and assume that $G$ is $B$-imperfect with respect to a total order $\pi$ on $V$, for some subset $B$ of $V$. Let $|V| = N$. Since $G$ is $B$-imperfect, we know that $G$ induced on $V \setminus B$ is a DAG. We shall refer to the elements in $V \setminus B$ as {\em good} elements. Let $\pi_g$ be the restriction of $\pi$ on the good elements. 
We begin with some key definitions: 

\begin{definition}[Balanced partition]
\label{def:balanced}
    Let $p$ be an element of a subset $S \subseteq V$. 
    A partition $L \cup R$ of $S \setminus \{p\}$ with respect to the pivot $p$ is said to be {\it balanced} if $\min(|L|, |R|) \geq \frac{|S|}{5}$.
\end{definition}

\begin{definition}[Support and loss of a sequence] Let $\sigma$ be a sequence on a subset of elements in $V$. 
    The {\em support} of the sequence $\sigma$, denoted $\supp(\sigma)$, is  defined as $\LCS(\sigma, \pi_g)$, i.e., the longest subsequence of good elements in $\sigma$ which appear in the same order as in $\pi$. If there are multiple choices for $\supp(\sigma)$,  we choose the one that is lexicographically smallest with respect to the indices in $\sigma$. Let $\loss(\sigma)$, the loss of $\sigma$,  be defined as the number of elements in $\sigma$ that are not in $\supp(\sigma)$.
\end{definition}

\begin{definition}[Concatenation Loss]
Consider two sequences $\sigma_1$ and $\sigma_2$. Let $\sigma$ be the sequence formed by the concatenation of $\sigma_1$ and $\sigma_2$. The {\em concatenation loss} of sequences $\sigma_1$ and $\sigma_2$, denoted  $\conloss(\sigma_1, \sigma_2)$, is defined as  $ |\supp(\sigma_1)| + |\supp(\sigma_2)| - |\supp(\sigma)|= \loss(\sigma) - \loss(\sigma_1) - \loss(\sigma_2)$. 
\end{definition}

Fix a subset $S \subseteq V$, and consider the recursive call $\RS(S)$ corresponding to $S$. Observe that the set $S$ changes during the run of this recursive call -- we shall use the index $t$ to denote an iteration of the {\bf while} loop in line~\ref{l:while} and use $S_t$ to denote the set $S$ during this iteration. Note that each iteration of the while loop either ends with the removal of a triangle or with recursive calls to smaller subproblems ($L$ and $R$).
We define several failure events with respect to an iteration of the while loop and show that these events happen with low probability: 
\begin{itemize}
\item {\bf Triangle Detection Test Failure}: This event, denoted ${\cal F}_1$, happens when $G[S_t]$ has at least $\frac{|S_t|^3}{24 k}$ triangles, but the {\bf for} loop in lines~\ref{l:triangle_test_begin}--\ref{l:triangle_test_end} does not find any triangle. 

\item {\bf Balanced Partition Failure}: This failure event, denoted ${\cal F}_2$, occurs when the procedure executes lines~\ref{l:endpart1}--\ref{l:endpart2} and then makes recursive calls to $L$ and $R$, but the partition $L \cup R$ is not a balanced partition of $S_t \setminus \{p\}$.  

\item {\bf Density Test Failure}: This event, denoted ${\cal F}_3$, happens when $|S_t|$ has at most $|S_t|/3$ bad elements, but we execute line~\ref{l:partition_test_end} in each of the $36 \log N$ iterations of the {\bf for} loop (line~\ref{l:no_triangle_found}). In other words, we are not able to find a good partition of the sampled $k'$ elements in any of the $36 \log N$ iterations of this {\bf for} loop. 
\end{itemize}
We now show that with high probability, none of the failure events happen. 
\begin{lemma}
    \label{lem:nofailure}
    Let $N$ be large enough (greater than some constant). Then
    $\Pr[{\cal F}_i] \leq \frac{1}{N^3} \forall i \in \{1, 2, 3\}$.
\end{lemma}
\begin{proof}
We first consider ${\cal F}_1.$ Suppose $S_t$ has at least  $\frac{|S_t|^3}{24 k}$ triangles. The probability that an iteration of the {\bf for} loop in lines~\ref{l:triangle_test_begin}--\ref{l:triangle_test_end} does not find a triangle is at most $1-\frac{1}{24k}$. Thus, the probability that none of the $72 k \log N$ iterations find a triangle is at most 
 $\left(1 - \frac{1}{24 k} \right)^{72 k \log N} \leq \frac{1}{N^3}.$

 We now analyze the event ${\cal F}_2. $ Suppose we choose a pivot $p$ in line~\ref{l:samplepivot} which is not balanced w.r.t. $S_t$. We need to argue that we shall execute line~\ref{l:partition_test_end} (which happens when $\min(|L_1|, |R_1|) \leq \frac{k'}{5} + \frac{k'}{40}$) with high probability. Let $L$ and $R$ denote the set of elements in $S_t \setminus \{p\}$ that are lesser than and greater than $p$, respectively. Assume $|L| < |S_t|/5$ (the other case is similar). Thus, the expected size of $L_1$ is at most $\frac{k'}{5}$. It follows from standard Chernoff bounds that the probability that $|L_1| > \frac{k'}{5} + \frac{k'}{40}$ is at most:
 $e^{\frac{-\left(\frac{1}{40}\right)^2 \cdot \frac{k'}{5}}{3}} \leq e^{-4 \log N} \leq \frac{1}{N^4}$.
Since we can execute line~\ref{l:samplepivot} at most $36 \log N$ times during a particular iteration $t$ of the {\bf while loop}, it follows by union bound that $\Pr[{\cal F}_2] \leq \frac{1}{N^3}.$ (for large enough $N$).

Finally, we consider the event ${\cal F}_3$. Suppose $S_t$ has at most $|S_t|/3$ bad elements, i.e., there are at least $2|S_t|/3$ good elements. Consider the middle (in the ordering $\pi_g$) $|S_t| / 6$ good elements of $S_t$. Any such element has at least $\frac{1}{2} \cdot (2|S_t| / 3 - |S_t| / 6) = |S_t| / 4 $ good elements on either side.  It follows that if the pivot $p$ is chosen among these $|S_t|/6$ elements, then the expected size of the sets $L_1, R_1$ defined in lines~\ref{l:L_1_defined}--\ref{l:partitionend} is at least $\frac{k'}{4} = (\frac{k'}{5} + \frac{k'}{40}) + \frac{k'}{40}$. Thus, for such a pivot, using the Chernoff bound, the probability of executing line~\ref{l:partition_test_end} (that is, $\min(|L_1|, |R_1|) \leq \frac{k'}{5} + \frac{k'}{40}$) is at most  $1/4$.
    Hence, the probability that we do not execute line~\ref{l:no_balanced_pivot} in a particular iteration of the {\bf for} loop in line~\ref{l:no_triangle_found} is at least $\frac{1}{6} \cdot \frac{3}{4} = \frac{1}{8}.$ Thus, $\Pr[{\cal F}_3] = $ the probability that we execute line~\ref{l:no_balanced_pivot} in each of the iterations of this {\bf for} loop is at most 
$\left(1 - \frac{1}{8} \right)^{36 \log N} \leq \frac{1}{N^3}.$
$\qed$
\end{proof}

\noindent
{\bf{Induction Hypothesis.}} Given a subset $S$ of $V$, let $\pi'(S)$ be the sequence generated by $\RS(S)$. Note that $\pi'(S)$ is a random sequence. Given integers $n$ and $b$, let $\cS(n,b)$ denote all subsets $S$ of size $n$ of $V$ which have  $b$ bad elements. Let $L(n,b)$ denote the maximum expected loss of the sequence output by~\Cref{algo:robust} when run on a subset $S \in \cS(n, b)$. More formally, 
$$  L(n,b) \coloneqq \max_{S \in \cS(n,b)} \E[\loss(\pi'(S))].$$
We now state the induction hypothesis that shall prove the main result~\Cref{thm:rsort}:
\begin{align}
    \label{eq:IH}
  L(n,b) \leq 3 b + c b \log n, \quad \forall n \in \mathbb{N}, 0 \leq b \leq n \leq N, \mbox{ where } c = \frac{\varepsilon}{\log N}
\end{align}
We prove the above result by induction on $n$. When $n=1$, the result follows trivially. Now assume it is true for all $L(n',b'),$ where $n' \leq n-1$. Now, we would like to show it for $L(n,b)$ for some given $b \leq n$. Fix a subset $S \in \cS(n,b)$. We need to show that 
\begin{align}
    \label{eq:IH1}
    \E[\loss(\pi'(S))] \leq 3b + cb \log n.
\end{align}
We first check an easy case. 

\begin{clm}
    \label{cl:trianglecase}
    Suppose $S$ has at least $\frac{n^3}{24k}$ triangles, where $k$ is as stated in line~\ref{l:choosek}, then $\E[\loss(\pi'(S))] \leq 3b + cb \log n$.
\end{clm}
\begin{proof}
    We know by~\Cref{lem:nofailure} that the failure event ${\cal F}_1$ happens with probability at most $1/N^3$. Thus, with probability at least $1-1/N^3$, the algorithm shall make a recursive call on a subset $S'$ obtained from $S$ by removing a triangle $(x,y,z)$ found in line~\ref{l:discardtriangle} (notice that, whenever we remove a triangle $\{x, y, z\}$, we start from scratch after setting $S \gets S \setminus \{x, y, z\}$). This triangle must contain at least 1 bad element. We can assume that it has exactly one bad element (otherwise, the induction hypothesis applied on $S'$ only gets stronger because the r.h.s. of~\Cref{eq:IH} is monotonically increasing with $b$). Thus, 
    $\E[\loss(\pi'(S))] \leq \left(1 - \frac{1}{N^3} \right) (L(n-3,b-1) + 3) + \frac{n}{N^3}$,
    where the second term on the r.h.s. appears because the loss can never exceed $n$. Applying induction hypothesis on $L(n-3, b-1)$, we see that the above is at most 
    $$3(b-1) + c(b-1)\log(n-3) + 3 + \frac{n}{N^3} \leq 3b + cb \log n - c \log (n-3) + \frac{n}{N^3} \leq 3b + cb \log n,$$
    where the second inequality uses $c = \frac{\epsilon}{\log N}$ and assumes that $N$ exceeds a large enough constant (for a fixed $\epsilon$) $\qed$ 
\end{proof}
Henceforth, assume that $S$ has at most $\frac{n^3}{24k}$ triangles. Further, if the algorithm finds a triangle in line~\ref{l:discardtriangle}, then the same argument as above applies. Thus, we can condition on the event that no triangles are found in the {\bf for} loop in lines~\ref{l:triangle_test_begin}--\ref{l:triangle_test_end}. We can also assume that $b \leq n/3$, otherwise the condition~\eqref{eq:IH1} follows trivially. Assume that ${\cal F}_3$ does not occur (we shall account for this improbable event later). In that case, we know that we shall find a balanced pivot $p$ of $S$ during one of the iterations of the {\bf for} loop in line~\ref{l:no_triangle_found}, and such that the condition in line~\ref{l:test} will not hold.
Let $L$ and $R$ be as defined in lines~\ref{l:endpart1}--\ref{l:endpart2}. We now give a crucial definition. 

\begin{definition}
    Let $p \in S$ be a pivot and $L$ and $R$ be the partition of $S \setminus \{p\}$ consisting of elements lesser than and greater than $p$, respectively. Define $M_p$ as: 
    $ M_p \coloneqq \max_{X, Y} \conloss(X,Y)$,
    where $X$ varies over all sequences of distinct elements of $L$ and $Y$ varies over all sequences of distinct elements of $R$. 
\end{definition}
Recall that, we say that $a$ is {\it less than} $b$, if $(a, b) \in E$. Using the definition of $\loss$, it is easy to verify that $M_p=0$ if $p$ is a good element. 
For an element $p \in S$, let $q_p$ denote the probability that no triangles are found in the {\bf for} loop in line~\ref{l:ktrials} conditioned on the fact that $p$ is chosen in the pivot in line~\ref{l:samplepivot} and the execution reaches the {\bf for} loop in line~\ref{l:ktrials}. We have the following key lemma: 
\begin{lemma}
    \label{lem:qp}
    For any $p \in S$,  $\left( 1 - \frac{2T_p}{n^2} \right)^k \leq q_p \leq \left( 1 - \frac{M_p^2}{4n^2} \right)^k,$ where $T_p$ denotes the number of triangles in $S$ containing $p$.
\end{lemma}
\begin{proof}
    We first prove the lower bound on $q_p$. Consider an element $p \in S$.  The probability that during an iteration of the {\bf for} loop in line~\ref{l:ktrials}, we pick the (unordered) triplet $\{x,y,p\}$ for a triangle $(x,y,p)$ is $\frac{2}{n^2}$. Since there are $T_p$ triangles containing $p$, the probability that we pick a triangle containing $p$ is at most $\frac{2T_p}{n^2}$. Thus, the probability that we do not pick any such triangle during the $k$ iterations of the {\bf for} loop in line~\ref{l:ktrials} is at least $\left( 1 - \frac{2T_p}{n^2} \right)^k.$

Now, we prove the upper bound. Assume $M_p > 0$; otherwise, the claim is trivial.  Let $L$ and $R$ denote the partition of $S \setminus \{p\}$ with respect to the pivot $p$. Let $X$ and $Y$ be sequences over subsets of $L$ and $R$ respectively such that $M_p = \conloss(X,Y)$. Recall that $\supp(X)$ (or $\supp(Y)$) is the longest common subsequence between $X$ (or $Y$) and the optimal sequence $\pi^\star_g$ on the good elements. In particular, $\supp(X)$ and $\supp(Y)$ consist of sorted good elements.
\begin{figure}[h]
    \centering
    % \resizebox{\columnwidth}{!}
    \begin{tikzpicture}

% Define positions of nodes (circular nodes only)
\foreach \i in {1,...,15} {
    \node[draw, circle, minimum size=4mm] (\i) at (\i*0.8, 0) {};
}

% Color the first 6 and last 6 nodes green
\foreach \i in {1,...,7} {
    \fill[green] (\i) circle[radius=2mm];
}
\foreach \i in {9,...,15} {
    \fill[green] (\i) circle[radius=2mm];
}

% Color the middle node black
\fill[red] (8) circle[radius=2mm];

% Labels for x, p, and y (placed below the nodes without enclosing them in circles)
\node[below=3mm] at (6) {\(x\)};
\node[below=3mm] at (8) {\(p\)};
\node[below=3mm] at (11) {\(y\)};

% Braces for supp(X) and supp(Y), pointing downwards
\draw[decorate,decoration={brace,mirror,amplitude=30pt}] (1.south west) -- (7.south east) node[midway,below=30pt] {\(\text{supp}(X)\)};
\draw[decorate,decoration={brace,mirror,amplitude=30pt}] (9.south west) -- (15.south east) node[midway,below=30pt] {\(\text{supp}(Y)\)};

% Braces for Mp/2 (pointing upwards)
\draw[decorate,decoration={brace,amplitude=15pt}] (5.north west) -- (7.north east) node[midway,above=15pt] {\(|X_1| = \frac{M_p}{2}\)};
\draw[decorate,decoration={brace,amplitude=15pt}] (9.north west) -- (11.north east) node[midway,above=15pt] {\(|Y_1| = \frac{M_p}{2}\)};

% Dashed curved arrow from x to y
\draw[dashed,->] (6.300) to[out=300,in=240] (11.240);
\foreach \i in {1,...,6} {
    \draw[solid,->] (\i.east) -- (\the\numexpr\i+1.west);
}
\foreach \i in {9,...,14} {
    \draw[solid,->] (\i.east) -- (\the\numexpr\i+1.west);
}

\end{tikzpicture}

    \caption{An arc from $x$ to $y$ is a contradiction to the concatenation loss being $M_p$. Thus, the tuple $(x,p,y)$ forms a triangle $x\rightarrow p \rightarrow y  \rightarrow x$.} 
    \label{fig:concat_loss}
    \end{figure}
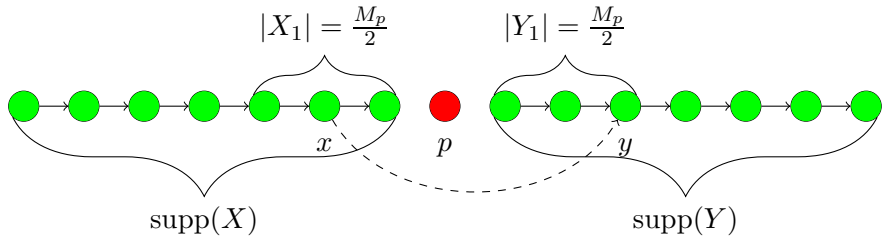
    Let $X_1$ denote the suffix of length $M_p / 2$ of $\supp(X)$, and $Y_1$ denote the prefix of length $M_p / 2$ of $Y$. We claim that for any $x \in X_1, y \in Y_1, (x,y,p)$ forms a triangle. Suppose not. Then, $x$ is less than $y$. But then, consider concatenating the prefix of $\supp(X)$ till $x$ and the suffix of $Y$ from $y$ (see \Cref{fig:concat_loss}). These sequence of good elements will be sorted and hence a subsequence of $\pi^\star_g$. But the length of this sequence is larger than  $\supp(X) + \supp(Y) - M_p. $ This implies that $\conloss(X,Y) < \supp(X) + \supp(Y) - (\supp(X) + \supp(Y)  - M_p) = M_p,$ a contradiction. Thus, we see that there are at least $\frac{M_p^2}{4}$ triangles involving $p$ and elements of $S \setminus \{p\}$. So, any particular iteration of the {\bf for} loop in line~\ref{l:ktrials} shall find such a triangle with probability at least $\frac{M_p^2}{4n^2}$. This proves the desired upper bound on $q_p$.
    $\qed$
\end{proof}

 Let $\Bal(S)$ denote the set of all balanced pivots of $S$, that is, the elements $p$ of $S$ for which the partition $L \cup R$ of $S \setminus \{p\}$ is balanced ($\min(|L|, |R|) \geq \frac{n}{5}$). For a pivot $p$, let $h_p$ denote the probability that, when $p$ is chosen as pivot (in line~\ref{l:samplepivot}),  we get a balanced partition $(L_1, R_1)$ of the sampled elements, that is, $\min(L_1, R_1) > \frac{k'}{5} + \frac{k'}{40}$ (and hence, line~\ref{l:partition_test_end} is not executed).  Recall that $b \leq \frac{n}{3}$. Hence, assuming ${\cal F}_3$ does not occur, there must be some pivot, among the (up to) $36 \log N$ sampled pivots, that passes the balanced partition test (i.e., line~\ref{l:partition_test_end} is not executed for this pivot). Let $\gamma_p$ denote the probability that this pivot is $p$. It is easy to see that $\gamma_p = \frac{h_p}{\sum_{u \in S} h_u}$.

\begin{lemma}
\label{lem:gamma_bound}
    Let $\Mid(S)$ denote the set of middle (in the ordering $\pi_g$) $|S|/6$ good elements of $S$. We have:
    \begin{itemize}
        \item[(i)] $\sum_{p \in S \setminus \Bal(S)} \gamma_p = \Pr[{\cal F}_2] \leq \frac{1}{N^3}$.
        \item[(ii)] For all $p \in \Mid(S)$, $\gamma_p \geq \frac{1}{2n} $.
        \item[(iii)] For all $p \in B$, $\gamma_p \leq \frac{12}{n}$.
    \end{itemize}
\end{lemma}

\begin{proof}
    The first statement follows from the definition of ${\cal F}_2$ and ~\Cref{lem:nofailure}. Now, notice that for $p \in \Mid(S), h_p \geq \frac{3}{4}$ (this was argued in the proof of~\Cref{lem:nofailure} when bounding the probability of $F_3$), and $\sum_{p \in S} h_p \leq n$. Hence, $\gamma_p \geq \frac{1}{2n}$ for all $p \in \Mid(S)$.
    For the third statement, note that $\sum_{p \in S} h_p \geq \sum_{p \in \Mid(S)} h_p \geq \frac{n}{6} \cdot \frac{1}{2} = \frac{n}{12}$. Also, $h_p \leq 1$  for all $b \in B.$ Hence, $\gamma_p \leq \frac{12}{n}$ for all $p \in B$. $\qed$

\end{proof}
 
Now, once a pivot $p$ passes the partition test, there are two possibilities: 
\begin{itemize}
\item[(a)] With probability $q_p$, we partition $S$ into $L$ and $R$ with respect to $p$ and recursively call $\RS(L)$ and $\RS(R)$. Let the output of these recursive calls be sequences $\sigma_L$ and $\sigma_R$, respectively (recall that the output of $\RS(L)$ or $\RS(R)$ can be a sequence on a subset of $L$ or $R$ respectively). Let $\sigma$ denote the concatenated sequence $(\sigma_L, p, \sigma_R)$.  If $p$ is a good element, then $M_p=0$, and hence $\loss(\sigma) = \loss(\sigma_1) + \loss(\sigma_2)$. If $p$ is a bad element, $\loss(\sigma)$ is same as the loss of the sequence $(\sigma_1, \sigma_2)$. The latter quantity, by the definition of $M_p$, is at most $\loss(\sigma_1) + \loss(\sigma_2) + M_p.$ Let $L$ and $R$ have $b_L$ and $b_R$ bad elements, and let their sizes be $n_L$ and $n_R.$ For, $p \in S \setminus \Bal(S), \E[\loss(\sigma)] \leq n$. For $p \in \Bal(S)$, 
    \begin{align*}
        \E[\loss(\sigma)] & \leq \E[\loss(\sigma_L)] + \E[\loss(\sigma_R)] + M_p \\
        & \leq 3b_L + cb_L \log n_L + 3b_R + cb_R \log n_R + M_p \\
        & \leq 3b + c b \log \max(n_L, n_R) + M_p  \\
        &\leq 3b + cb \log n - c b \log(5/4)  +M_p,
    \end{align*}
where the second inequality follows from the induction hypothesis, and the last inequality uses $\max(n_L, n_R) \leq \frac{4n}{5}$.

\item[(b)] With probability $(1-q_p)$, a triangle is found in lines~\ref{l:ktrials}--\ref{l:triangle_check_end}. In this case, we effectively recurse on a subset $S'$ of $S$, where $S'$ has $n-3$ elements and at most $b-1$ bad elements. In this case, the induction hypothesis implies that the expected loss of the sequence returned by $\RS(S)$ is, at most
$$ \E[\pi'(S')] +3 \leq 3(b-1) +  c(b-1) \log (n-3) + 3 \leq 3b + c b \log n$$ where the expectation is over the choice of $S'$ and the $\RS(S')$ procedure. 
\end{itemize}
  Putting everything together, we see that (conditioned on ${\cal F}_3$ not happening), $\E[\loss(\pi'(S)) | {\cal F}_3 \text{does not happen}]$ is at most (recall that $B$ is the set of bad elements):
 \begin{align}
     &  3 b+ cb \log n - c b \log(5/4) \sum_{p \in \Bal(S) \setminus B} q_p \gamma_p  + \sum_{p \in S \cap B} q_p \gamma_p \cdot M_p   + \sum_{p \in S \setminus \Bal(S)} \gamma_p n\notag \\
     \leq &  3 b+ cb \log n - c b \log(5/4) \sum_{p \in \Mid(S)} \frac{1}{2n} q_p  + \sum_{p \in S \cap B} \frac{12}{n} q_p \cdot M_p +  \frac{n}{N^3} \notag \\
      \leq & \ 3b + cb \log n - c b \log(5/4) \sum_{p \in \Mid(S)} \frac{1}{2n} \left( 1 - \frac{2T_p}{n^2} \right)^k   + \frac{12}{n} \sum_{p \in S \cap B}  M_p \left(1 - \frac{M_p^2}{4n^2}\right)^k + \frac{1}{N^2} \label{eq:calc}
    \end{align}
    where the first inequality uses the observation that $\Mid(S) \subseteq \Bal(S) \setminus B$, and utilizes the bounds on $\gamma_p$ derived in ~\Cref{lem:gamma_bound}. The second inequality uses the bounds on $q_p$ derived in ~\Cref{lem:qp}.
     Now, observe that, the expression $x (1 - x^2 )^k$ over $x > 0$, is maximized at $x = \frac{1}{\sqrt{2k + 1}}$, and hence,     
     
    $$M_p \left(1 - \frac{M_p^2}{4n^2}\right)^k =  2n \frac{M_p}{2n} \left(1 - \left(\frac{M_p}{2n}\right)^2\right)^k \le \frac{2n}{\sqrt{2k + 1}} \left(1 - \frac{1}{2k + 1} \right)^k \le \frac{2n}{\sqrt{2k + 1}}$$
    
    Substituting this in~\eqref{eq:calc}, we see that $\E[\loss(\pi'(S)) | {\cal F}_3 \text{does not happen}]$ is at most:
    \begin{align*}
    % \label{eq:calc2}
         &3b + cb \log n  - cb \log (5/4) \sum_{p \in \Mid(S)} \frac{1}{2n} \left( 1 - \frac{2T_p}{n^2} \right)^k  + \frac{12 |S \cap B|}{n} \frac{2n}{\sqrt{2k + 1}} + \frac{1}{N^2}\\
         & \leq 3b + cb \log n  - cb \log (5/4) \cdot \frac{1}{2n} \cdot \sum_{p \in \Mid(S)}  \left( 1 - \frac{2T_p}{n^2} \right)^k  + \frac{12 \sqrt{2} b}{\sqrt{k}} + \frac{1}{N^2}
    \end{align*}

After Claim \ref{cl:trianglecase}, we had assumed that $\sup_p T_p \leq \frac{n^3}{24 k}$. Therefore, there can be at most $n/12$ elements for which $T_p > \frac{n^2}{2k}$. It follows that there are at least $n/6 - n/12 = n/12$ elements $p$ of $\Mid(S)$ for which $T_p \leq \frac{n^2}{2k}$. Therefore, the expression above is, at most: 
\begin{align*}
    & 3b + cb \log n  - cb \log (5/4) \cdot \frac{1}{2n} \cdot \frac{n}{12} \left( 1 - \frac{1}{k} \right)^k  + \frac{12 \sqrt{2} b}{\sqrt{k}} + \frac{1}{N^2} \\
    & \leq 3b + cb \log n - \frac{c}{96} \, \log(5/4) b + \frac{12 \sqrt{2} b}{\sqrt{k}} + \frac{1}{N^2}\\
    & \leq 3b + cb \log n - \frac{b \epsilon}{\log N} \left( \frac{\log(5/4)}{96} - \frac{12 \sqrt{2}}{10000}\right) + \frac{1}{N^2}\\
    & \leq 3b + cb \log n - \frac{1}{N^2}
\end{align*}
where the first inequality uses the fact that $k \geq 2$ and hence, $(1-1/k)^k \geq 1/4$, the second one uses the fact that  $c = \frac{\epsilon}{\log N}$, and the last inequality assumes that $N$ exceeds a large enough constant (for a fixed $\epsilon$), and $b \geq 1$ (for $b = 0$, $L(n, b) = 0$ trivially holds). 
Now, 
\begin{align*}
    \E[\loss(\pi'(S))] = &\Pr[{\cal F}_3 \text{ does not happen}] \cdot \E[\loss(\pi'(S)) | {\cal F}_3 \text{ does not happen}]\\
    & + \Pr[{\cal F}_3 \text{ happens}] \cdot \E[\loss(\pi'(S)) | {\cal F}_3 \text{ happens}] \\
    & \leq 1 \left(3b + cb \log n - \frac{1}{N^2} \right) + \frac{1}{N^3} \cdot n\\
    & \leq 3b  + cb \log n
\end{align*}
Finally, using $c = \frac{\epsilon}{\log N}$, we see that the expected loss is at most $(3+\epsilon)b$. 
This proves the statement about the quality of the output permutation in~\Cref{thm:rsort}. It remains to calculate the expected number of queries by the algorithm. 
\begin{lemma}
    \label{lem:countqueries}
    Conditioned on the failure event ${\cal F}_2$ not happening during any of the recursive calls of $\RS$, the number of queries made by the algorithm is $O\left(\frac{N \log^3 N}{\epsilon^2}\right)$. The same bound holds for the running time of the algorithm. 
\end{lemma}

\begin{proof}
    It is easy to verify that each iteration of the {\bf while loop} (\ref{l:while}) takes $O(\frac{\log^3 N}{\epsilon^2})$ time, and either finds a triangle or divides the problem into two smaller subproblems (in this case taking an additional $O(|S|)$ time). Since ${\cal F}_2$ does not happen, we have that the subproblem sizes are at most $O(\frac{4}{5} |S|)$.  Hence, the time complexity recursion (which subsumes the query complexity) is:
    $$T(n) = \max\left (T(n-3), \max_{n/5 \leq n_1 \leq 4n/5} (T(n_1) + T(n - 1 - n_1) + O(n)) \right) + O\left(\frac{\log^3 N}{\epsilon^2}\right).$$
    It is easy to inductively prove that $T(n) = O\left(n \log n + n \frac{\log^3 N}{\epsilon^2}\right).$
    $\qed$
\end{proof}

The probability of ${\cal F}_2$ happening during an iteration of the {\bf while loop} (\ref{l:while}) is at most $1/N^3$ (~\Cref{lem:nofailure}). Since there are up to $N$ iterations of the while loop overall, using union bound, the probability that ${\cal F}_2$ never happens is at least $1 - 1/N^2$. Also, the worst-case running time (and the query complexity) of $\RS$ is $O(N^2)$. Hence, the expected running time (also the expected query complexity) of the algorithm is at most:
$$\left(1-1/N^2\right) \cdot O\left(\frac{N \log^3 N}{\epsilon^2}\right) + \frac{1}{N^2} \cdot N^2 = O\left(\frac{N \log^3 N}{\epsilon^2}\right).$$

Note that, when proving $L(n, b) \le 3b + cb \log n$, where $c = \frac{\epsilon}{\log N}$, we assumed that $N$ is large enough at certain places. It can be easily verified that $N = \Omega (\frac{1}{\epsilon^{2/3}})$ suffices in all these places. To extend the small loss guarantee to smaller $N$, we can simply tweak our algorithm to run the triangle removal algorithm if $N = O(\frac{1}{\epsilon^{2/3}})$, resulting in a loss of at most $3 b$, with a run-time of $O(\frac{1}{\epsilon^2})$.
This completes the proof of~\Cref{thm:rsort}.

\section{Linear time $(2-\delta)$-approximation algorithm for Ulam-$k$-Median}
\label{sec:ulamoutline}
\newcommand{\Bad}{{\mathsf{Bad}}}
\newcommand{\Good}{{\mathsf{Good}}}
In this section, we give the main ideas behind the linear time approximation algorithm for the Ulam-$k$-Median problem. Details are deferred to~\Cref{sec:ulamdetails}. Our algorithm builds on the algorithm of~\cite{cdk23}. 
They gave a $(2-\delta)$-approximation algorithm for the Ulam-$1$-Median problem for a constant $\delta > 0$. 
We describe their approach briefly. Consider an input for the Ulam-$1$-Median problem given by a set $S$ of $n$ permutations of $[d]$. 
\begin{itemize}
\item[(i)] They show that there are five special input permutations, say $\sigma_1, \ldots, \sigma_5$, which (as the next step describes) suffice to find a permutation close to the optimal permutation $\sigma^*$. In their algorithm, they try out all $O(n^5)$ possibilities for finding these five input permutations. 
\item[(ii)] Recall that $\LCS(\sigma, \sigma^*)$ denotes the longest common subsequence between $\sigma$ and $\sigma^*$. Define $I_\sigma$ denote the symbols in $[d]$ which do not appear in $\LCS(\sigma, \sigma^*)$. Order the strings $\sigma_1, \ldots, \sigma_5$ in ascending order of $|I_{\sigma_i}|$ values. Then, these permutations satisfy the following properties: 
\begin{itemize}
    \item[(P1)] For every $j > i$, $|I_{\sigma_i} \cap I_{\sigma_j}| \leq \varepsilon \cdot |I_{\sigma_i}|,$ where $\varepsilon > 0$ is a small enough constant. 
    \item[(P2)] $|I_{\sigma_4}| \leq (1+\alpha)\frac{OPT}{n}$, for some small constant $\alpha$.
\end{itemize}
\item[(iii)] Let $\Good$ denote the set of symbols in $[d]$ which are present in at most one of the sets $I_{\sigma_i}, i \in [5]$. Let $\Bad$ denote the remaining symbols, i.e., $\Bad \coloneqq [d] \setminus \Good.$  Property~(P1) shows that the number of symbols in $\Bad$ is at most 
$$
 4 \veps |I_{\sigma_1}| + 3 \veps |I_{\sigma_2}| + 2 \veps |I_{\sigma_3}| + \veps |I_{\sigma_4}| \leq 10 \veps |I_{\sigma_4}|.
$$
Using $\veps \leq 1/100$ and property~(P2), it can be shown that the number of bad symbols is at most $d/10$.
\item[(iv)] Let $a, b \in \Good$. It follows that at least three out of the five permutations $\sigma_1, \ldots, \sigma_5$ agree with $\sigma^*$ on the relative ordering of $a$ and $b$. Thus, they create an instance of the \robust problem where the outcome of a query on a pair $(a,b)$ is given by the majority vote among $\sigma_1, \ldots, \sigma_5$. It follows that when $a, b$ both belong to $\Good$, then the outcome agrees with the optimal permutation. 
\item[(v)] It remains to solve the \robust problem.~\cite{cdk23} use the following approach: query all pairs $(a,b)$ to a get a directed graph $G$. Whenever $G$ has a directed triangle, remove all three vertices in the triangle. Clearly, at least one of the three vertices in this triangle would belong to $\Bad$, and hence, this upper bounds the number of removed points by $3 |\Bad|$. However, this procedure takes $O(d^3)$ time. 
\end{itemize}

We improve on the above approach as follows: (a) In the first step above, instead of trying out all subsets of size 5 of the input, we show that there is a small (constant size) random sample of $S$ that is guaranteed to contain five permutations satisfying (P1) and (P2) with high probability, and (b) we can use our linear time algorithm for \robust in step~(v) above to improve the running time dependence on $d$ to linear, (c) Finally, we show that resolving the $1$-median through a random sampling approach leads directly to a corresponding algorithm for the $k$ median problem. We now outline each of these steps. 

\subsection{Random Sampling Lemma for Ulam-$1$-median}
Recall that for any permutation $\sigma$, we define $I_{\sigma}$ to be the set of symbols that are misaligned with the optimal permutation, i.e., $I_{\sigma}$ contains those symbols that are not in $\LCS(\sigma,\sigma^{*})$, and hence, $|I_{\sigma}| = \Delta(\sigma, \sigma^*)$.
The algorithm of Chakraborty et al.~\cite{cdk23} performs a case analysis to find a good center (we say `good' in the sense of breaking the 2-approximation barrier) for the input set $S$. 
Let $\sigma_c$ denote the permutation in $S$ closest to $\sigma^*$. 
The analysis in \cite{cdk23} breaks into the following two initial cases: (i) $\Delta(\sigma_c, \sigma^*) \leq (1-\alpha)\frac{OPT}{n}$, and (ii) otherwise. In case (i), we can show that $\sigma_c$ is a good center. This means that there {\bf exists} a good center within the input set $S$. 
This existence is sufficient for the algorithm in~\cite{cdk23} since it will go over all combinations of permutation in $S$ to find a good center. However, our aim is to show that a random subset is sufficient to locate a good center. So, we need to consider the following cases instead: (i) there are many points in $S$ near $\sigma^*$, and (ii) there are very few such points. In case (i), a random subset will contain a good center. For (ii), we consider further cases as in \cite{cdk23}. 
Even though case (i) is more straightforward than the subsequent subcases of case (ii), it exhibits our approach.

The key insight of \cite{cdk23} is to use the fact that for two permutations $\sigma_1$ and $\sigma_2$, $\Delta(\sigma_1, \sigma_2) \leq |I_{\sigma_1}| + |I_{\sigma_2}| - |I_{\sigma_1} \cap I_{\sigma_2}| = \Delta(\sigma_1, \sigma^*) + \Delta(\sigma_2, \sigma^*) - |I_{\sigma_1} \cap I_{\sigma_2}|$. This slightly stronger triangle inequality is used cleverly to break the 2-approximation barrier. Their analysis can be subdivided into the following two high-level cases: (a) a good number of permutations are clumped together (i.e., are pair-wise close to each other), and (b) there are at least five permutations that are reasonably spread apart. The measure of being spread/clumped is in terms of the number of common misaligned symbols. More precisely, two permutations $\sigma_1$ and $\sigma_2$ are said to be close if $|I_{\sigma_1} \cap I_{\sigma_2}|$ is sufficiently small. For case (a), \cite{cdk23} argues that there exists a good permutation within the clump. 
For our random subset version, we need to consider potential clumps centered around various permutations. If there are too many such permutations, then a random subset will contain one. On the other hand, if there are too few, then there will be sufficiently many permutations that are spread apart, and hence, a random subset will also have enough permutations that are spread apart. So, we get that in a {\em constant size} random subset of $S$, either there is a permutation that is a good center, or there are five permutations that are reasonably separated (in fact, we can show that there are five permutations that satisfy (P1) and (P2)).\footnote{Chakraborty et al.~\cite{cdk23} also show a similar sampling lemma, however, they required a sample of size $O(\log{n})$, which leads to an additional $(\log{n})^{O(k)}$ factor in the running time for the $k$-median problem, which we wanted to avoid.}
In the latter case, the analysis is the same as that in \cite{cdk23}, except that we use our robust sorting algorithm to find a good center from five permutations instead of the $O(d^3)$ algorithm in \cite{cdk23}. We discuss this next.

\subsection{Using the linear time algorithm for \robust}
Suppose we have five input permutations $\sigma_1, ..., \sigma_5$ that satisfy properties (P1) and (P2) as described in the second step of our approach given at the beginning of this section. 
Let us construct a tournament graph $T = ([d], E)$ where the nodes are the $d$ symbols, and there is a directed edge from symbol $a$ to symbol $b$ if and only if $a$ comes before $b$ in at least three out of five permutations $\sigma_{1}, ..., \sigma_{5}$. The goal is to use $T$ to find an ordering such that the ordering of ``most" symbols matches that in $\sigma^*$. This is where our algorithm deviates from that of \cite{cdk23}. 
The algorithm in \cite{cdk23} achieves this by finding and removing cycles of length three, thereby removing all bad symbols since every such cycle must have a bad symbol. Once all the bad symbols have been removed, the remaining graph has only good symbols. Since the relative ordering of every pair of good symbols matches that in $\sigma^*$, the correct ordering of the remaining symbols is obtained. The time cost of this algorithm in \cite{cdk23} is $O(d^3)$.\footnote{One can also use the 2-factor approximation algorithm of Lokshtanov et al.~\cite{fvst} for the Feedback Vertex Set problem on Tournament graphs (FVST). However, the running time of their algorithm, $O(d^{12})$, is even worse than cubic.}
We give an algorithm with time linear in $d$ to achieve the same goal using our {\em robust sorting} framework. 
When applying our robust sorting framework, total ordering is per the ordering of symbols in the permutations $\sigma^*$, and the comparison operator {\tt $COMPARE_{\sigma_{1}, ..., \sigma_{5}}(a, b)$} checks if at $a$ is before $b$ in at least three out of five of the permutations $\sigma_{1}, ..., \sigma_{5}$.
The remaining aspects, including the good and bad symbols, perfectly align with our current goal. 
That is, the comparison operator gives correct relative ordering (in terms of being consistent with $\sigma^*$) for a pair in which both symbols are good, but if even one symbol in the pair is bad, then no such guarantees exist.
Note that the running time of our algorithm {\tt ROBUST-SORT} is nearly linear in $d$ (ignoring polylogarithmic factors in $d$). 
Using Theorem~\ref{thm:rsort}, we get that that {\tt ROBUST-SORT} returns a permutation $\tilde{\sigma}$ with the property that $\lcs(\sigma^*, \tilde{\sigma}) \geq d - 4 |Bad|$. 
The rest of the analysis is again similar to \cite{cdk23}. We have:
$$
\Delta(\tilde{\sigma}, \sigma^*) = d - \lcs(\tilde{\sigma}, \sigma^*) \leq 4 |Bad| \stackrel{(P1)}{\leq} 40 \veps |I_{\sigma_4}| \stackrel{(P2)}{\leq} 40 \veps (1 + \alpha) \frac{OPT}{n}.
$$
Using this, we can now calculate the cost with respect to the center $\tilde{\sigma}$.
$$
Obj(S, \tilde{\sigma}) = \sum_{x \in S} \Delta(x, \tilde{\sigma}) 
\stackrel{(triangle)}{\leq} \sum_{x \in S} \left( \Delta(x, \sigma^*) + \Delta(\tilde{\sigma}, \sigma^*)\right)
\leq \left( 1 + 40 \veps (1+\alpha)\right) \cdot OPT \\
\leq (1.999) \cdot OPT.
$$
The last inequality is by a specific choice of $\veps$ and $\alpha$. Combining the results of this and the previous subsection, we obtain the following sampling lemma that may be of independent interest.

\begin{lemma}[A sampling lemma]
Let $S$ denote any input for the Ulam-$1$-median problem, a subset of permutations of $1, 2, ..., d$. There exists a constant $\eta$ and an algorithm {\tt ULAM1} such that when given as input a set $X \subset S$ consisting of $\eta$ permutations chosen uniformly at random from $S$, {\tt ULAM1} outputs a list of $O(\eta)$ centers such that with probability at least 0.99, at least one center $\tilde{x}$ in the list satisfies $Obj(S, \tilde{x}) \leq (2-\delta) \cdot OPT$, where $\delta$ is a small but fixed constant. Moreover, the running time of {\tt ULAM1} is $\tilde{O}(d)$.
\end{lemma}
The details of the proof can be found in ~\Cref{sec:ulamdetails}.
Next, we see how to use the above lemma to obtain a linear time algorithm for the Ulam-$k$-median problem.

\subsection{From Ulam-$1$-median to Ulam-$k$-median}
We use the algorithm {\tt ULAM1} in the sampling lemma above as a subroutine to design an algorithm for the Ulam-$k$-median problem.
This algorithm is based on the $D$-sampling technique.
The algorithm and its analysis follow the $D$-sampling based $(1+\veps)$-approximation results of \cite{jks} for the $k$-means problem.
For readers familiar with the $D$-sampling-based algorithms of \cite{jks}, our approach can be summarised in the following manner: {\it the high-level algorithm description and its analysis are the same as the algorithm for the $k$-means problem in \cite{jks}, except that here we use our sampling Lemma above, instead of the Inaba's sampling lemma~\cite{inaba} that is used for the $k$-means problem.}
Here is a short summary for the unfamiliar reader.
Let $S_1, ..., S_k$ be the dataset partition that denotes the optimal $k$ clustering, and let $\sigma^*_1, ..., \sigma^*_k$ denote the optimal centers, respectively. Let us try to use {\tt ULAM1} to find good centers for each of $S_1, ..., S_k$. We would need $\eta$ uniformly sampled points each from $S_1, ..., S_k$. 
The issue is that the optimal clustering $S_1, ..., S_k$ is not known.
If the clusters were balanced, i.e., $|S_1| \approx |S_2| \approx ... \approx |S_k|$, then uniformly sampling $O(\eta k)$ points from $S$ and then considering all possible partitions of these points would give the required uniform samples $X_1, ..., X_k$ from each of the optimal clusters. We can then use {\tt ULAM1($X_i$)} to find good center candidates for $S_i$. 
In the general case, where the clusters may not be balanced, we use the $D$-sampling technique to boost the probability of sampling from small-sized optimal clusters, which may get ignored when sampling uniformly at random from $S$. We defer the details to ~\Cref{sec:ulamdetails}.

\section{Conclusion and open problems}
We give an algorithm for robust sorting. For a total order $\pi$ on a set $V$ of elements, we are given an imperfect comparison operator that behaves in the following manner: There is an (unknown) subset $B \subset V$ such that for every pair $a, b \in V\setminus B$, the comparator is consistent with the total order $\pi$, but if even one of $a, b$ is from $B$, then the comparator can give an arbitrary (but deterministic) response. We give an algorithm that outputs an ordering $\pi'$ such that the order of at least $|V| - (3+\epsilon) \cdot |B|$ elements are consistent with $\pi$ (i.e., $LCS(\pi, \pi') \geq |V| - (3+\epsilon)|B|$). Our algorithm runs in time $\tilde{O}(|V|/\epsilon^2)$. This means that it does not compare every pair of elements. 
Even though $3 + \epsilon$ was sufficient for the Ulam median application, it is natural to ask whether the factor $3 + \epsilon$ can be improved. More concretely, we identify the following open problem: 

\begin{quote}
    {\em What is the best constant $\alpha$, such that there exists a randomized algorithm that asks $\tilde{O}(|V|)$ queries in expectation, and outputs an ordering $\pi'$ with $\E[\LCS(\pi, \pi')] \geq |V| -  \alpha |B|$? In particular, as our algorithm provides a $(3+\eps)$-approximation, it would be interesting to see if $\alpha \leq 3$.}
\end{quote}

The above open question restricts the number of queries to $\tilde{O}(|V|)$. However, the problem remains interesting even if there is no bound on the number of queries. In particular, we ask the following open question:
    \begin{quote}
        
     {\em What is the best constant $\beta$, such that there exists a polynomial time randomized algorithm with no bound on the number of queries, that outputs an ordering $\pi'$ with $\E[\LCS(\pi, \pi')] \geq |V| -  \beta |B|$? In particular, as a reduction to FVST yields a $3$-approximation to \robust, it would be interesting to see if $\beta < 3$.}
    \end{quote}
% \end{enumerate}

\noindent
{\bf{Discussion.}} On the applied side, our robust sorting model has potential applications in certain problems involving ranking based on the subjective opinions of experts. The element of subjectivity makes ranking very tricky. A few examples are ranking Gymnastics, ice skating, and artistic performances. More serious scenarios include ranking organ recipient lists, such as a kidney transplant waitlist, where moral considerations may be required to break ties after objective processing. When subjectivity is involved in evaluation, it is harder for any expert to rank all the candidates than to answer pairwise comparison queries such as ``{\em whom do you prefer between the two given candidates?}" The relevant question is whether we can use such comparison queries to construct a reasonable ranking. The issue in such a process is that finding an expert without an explicit or implicit bias towards/against particular individuals is hard. The typical solution used is to take the opinions of multiple experts with the assumption that these individual biases should cancel out when considering the majority opinion. However, this may not be true, and the bias towards/against a small set of (unknown) individuals may remain. Let us call this set of individuals $B$ and let the remaining set of candidates be $G$. So, what we have is pairwise preference information that is consistent with the true ranking of individuals in $G$, but the comparison queries are arbitrary when the pair contains individuals from the biased set $B$. In such a crowdsourced setting, every pairwise comparison comes with a cost, and hence, the total number must be minimized (e.g., asking a judge to compare every pair of candidates may be too cumbersome). Furthermore, it can be argued that a sublinear number of queries is not sufficient to draw any reasonable ranking list. So, our nearly linear query algorithm may be an effective option in drawing out an initial list on which further processing could be done.

%%
%% Bibliography
%%

%% Please use bibtex, 
\addcontentsline{toc}{section}{References}
\bibliographystyle{alpha}
\bibliography{main}

\newpage

\appendix

\section{The Ulam-$k$-Median Problem}
\label{sec:ulamdetails}
In this section, we give a complete discussion on the Ulam $k$-median problem and our linear time approximation algorithm for this problem. 
The Ulam $k$-median problem is a special case of the $k$-median problem with which we start our discussion.
The $k$-Median problem on a general metric space $(\X, D)$ is defined in the following manner -- {\it Given a client set $S \subseteq \X$, a facility set $F \subseteq \X$, and an integer $k>0$, find a subset $C \subseteq F$ (called centers) such that $|C| = k$ and the cost function, $Obj(S, C) \coloneq \sum_{s \in S} \min_{c \in C} D(s, c)$, is minimized.} The problem size is denoted by $n = |S| + |F|$. It is known that the above metric $k$-median problem is hard to approximate within a factor of $(1+2/e)$ \cite{gk99}. On the upper bound side, various constant-factor approximation algorithms are known. The current best polynomial time approximation guarantee for the problem is $(2.675+\varepsilon)$~\cite{byrka17}. The problem has also been studied in the parameterized setting with the output size, $k$, as the parameter. Fixed Parameter Tractable (FPT) approximation algorithm (i.e., an algorithm with a running time $f(k, \veps) \cdot n^{O(1)}$) with a guarantee $(1+2/e + \veps)$, which closely matches the lower bound, is known~\cite{cgkll19}. 

The Ulam $ k$-median problem is simply the $ k$-median problem defined over the Ulam metric $(\Pi^d, \Delta)$, where $\Pi^d$ is the set of all permutations of $1, 2, 3, ..., d$, which can also be seen as all $d$-length strings with distinct symbols from set $\{1, 2, ..., d\}$.
The distance function $\Delta$ is defined as $\Delta(x, y) \coloneq d - \lcs(x, y)$, where $\lcs(x, y)$ denotes the length of the {\em longest common subsequence} of permutations $x, y \in \Pi^d$. For example, $\Pi[3] = \{123, 132, 213, 231, 312, 321\}$ and $\Delta(123, 213) = 3 - \lcs(123, 213) = 3 - 2 = 1$. 
Let us highlight some of the key aspects of the Ulam $k$-median problem where the computational complexity discussions deviate from those for the $k$-median problem. In the Ulam $k$-median problem, the facility set $F$ contains all the elements of the metric space, i.e., $F = \Pi^d$. So, $|F| = |\Pi^d| = (d!)$. 
In the $k$-median problem, the input size is $(|S| + |F|)$, and hence an efficient algorithm is one that runs in time polynomial in $(|S|+|F|)$. In the case of the Ulam $k$-median problem, the input is the set $S$ of $n$ strings from $\Pi^d$. The space required to describe such an input is $n d$ (assuming $\log{d}$ word size). 
So, an efficient algorithm for the Ulam $k$-median problem is one that runs in time polynomial in $n d$.
In this work, we will give a nearly linear time FPT algorithm (with $k$ as the parameter), meaning that our algorithm's running time should be $ \tilde{O}(f(k) \cdot nd)$.

We start with the Ulam $1$-median problem, also known as the {\em median string} problem for the Ulam metric. The problem is to find a string $s^* \in \Pi^d$ such that $Obj(S, s^*) \coloneqq \sum_{s \in S} \Delta(s, s^*)$ is minimized. The exhaustive search algorithm runs in time $O(n \cdot (d!))$ which is prohibitive. There is a $2$-factor approximation algorithm that simply picks the least cost string from the given dataset $S$. The approximation analysis follows from the triangle inequality. Let $\bar{s} \in S$ be the closest string in $S$ from $s^*$. Then we have:
$$
Obj(S, \bar{s}) = \sum_{s \in S} \Delta(s, \bar{s}) = \sum_{s \in S} (\Delta(s, s^*) + \Delta(s^*, \bar{s})) \leq 2 \cdot \sum_{s \in S} \Delta(s, s^*) = 2 \cdot OPT.
$$
Note that the above algorithm uses only the triangle inequality and does not use any other properties specific to the Ulam metric. 
This 2-factor approximation for 1-median also extends to the $k$-median setting since one can pick the best subset of size $k$ from $S$ in time $\Omega(n^k)$. Breaking the approximation guarantee of $2$ for specific metrics, such as the Ulam metric, was open for a long time. Recently, a sequence of works~\cite{cdk21,cdk23} breaks this 2-factor barrier for the Ulam $k$-median problem. The algorithm is an FPT algorithm that gives an approximation guarantee of $(2 - \delta)$ for a very small but fixed constant $\delta > 0$. The running time of the FPT algorithm is $(k \log{nd})^{O(k)} nd^3$. 
Note that this parameterized algorithm is slightly superlinear in $n$ and cubic in $d$.
This contrasts the several FPT approximation algorithms for the $k$-median/means problems in the literature~\cite{kss,jks,bjk18,bgjk20,gj20} that have linear running time in the input size. For example, the running time of the FPT $(1+\veps)$-approximation algorithms for the Euclidean $k$-median/means problems in \cite{kss,jks,bjk18,bgjk20}  is linear in $nd$, where $d$ is the dimension. Similarly, the running time for the FPT constant factor approximation algorithms~\cite{gj20} in the metric setting is linear in $n$. The starting point of this work was the following question: 
\begin{quote}
{\it Is it possible to break the barrier of 2-approximation for the Ulam $k$-median problem using an FPT algorithm with a {\bf linear} running time in the input size, $nd \log{d}$?}
\end{quote}
We needed two sets of ideas, one for reducing the time dependency on $n$ from superlinear to linear and another for reducing the dependency on $d$ from cubic to $\tilde{O}(d)$. 

\paragraph{\bf For linear dependence on $n$}: For improving the dependency on $n$, we use the $ D$-sampling-based ideas that have been used in the past (e.g., \cite{jks,bjk18}) to remove the $(\log{n})^k$ factors from the running time and make the dependency linear in $n$. 
Significant technical challenges existed in extending the $D$-sampling techniques to the Ulam $k$-median setting.
For example, the linear time algorithms for $k$-means/median~\cite{jks,bjk18} crucially use a simple sampling-based subroutine for the 1-median problem where one can find a good center for a set of points from a small randomly sampled subset. The algorithm of \cite{cdk23} also uses a subroutine for the 1-median problem. However, the guarantee of finding a good center from a small randomly sampled subset is not shown in \cite{cdk23}. What is shown is the guarantee that five points always exist in the dataset from which a good center can be found. This means that one must try $O(n^5)$ possibilities for the 1-median setting and $O(n^{5k})$ possibilities for the k-median setting. The number of possibilities can be reduced using a {\em coreset} in place of the dataset. However, this introduces a dependence on $(\log{n})^{O(k)}$ since a coreset size has a logarithmic dependency on $n$. We remove this dependence by showing that one can find a good center from a small randomly sampled subset. This enables the framework of \cite{jks,bjk18} to be applicable in the setting of the Ulam $k$-median problem. We give the details in the following subsections.

\paragraph{\bf For linear dependence on $d$}: Interestingly, the quest to reduce the dependency from cubic in \cite{cdk23} to almost linear in $d$ gave rise to a fundamental problem on {\em robust sorting} (or sorting using an imperfect comparator) that we think could have independent applications. Hence, we decided to make the robust sorting framework and algorithm the main subject of discussion in this paper. 
%The linear-time FPT algorithm and analysis details are in the following subsections. 
Let us outline the key aspects of improving the dependency on $d$ using our robust sorting framework. 
We must first outline the $(2-\delta)$-approximation algorithm of \cite{cdk23} at a high level. Focusing on the 1-median problem will be sufficient since the algorithm for $k$-median is through its extension.
Consider the 1-median problem -- {\it given a set $S \subset \Pi^d$ of size $n$, find $s^* \coloneqq \argmin_{s \in \Pi^d} \sum_{x \in S} \Delta(x, s)$.}
For any string $s$, let $slcs(s, s^*)$ denote the subset of symbols included in the longest common subsequence of $s$ and $s^*$ and let $I_{s} \coloneqq [d]\setminus slcs(s, s^*)$ be the symbols that are not part of the longest common subsequence, which can also be thought of symbols in $s$ that are {\em misaligned} with the optimal string $s^*$.
The algorithm of Chakraborty et al.~\cite{cdk23} shows that there always exists a set of five strings $s_1, s_2, s_3, s_4, s_5 \in S$ with the following three properties: 
\begin{enumerate}
\item[(A)] $|I_{s_1}| \leq |I_{s_2}| \leq |I_{s_3}| \leq |I_{s_4}| \leq |I_{s_5}|$
\item[(B)] For every pair $j > i$, $|I_{s_i} \cap I_{s_j}| \leq \veps \cdot |I_{s_i}|$.
\item[(C)] $|I_{s_4}| \leq (1+\alpha)\frac{OPT}{n}$, for some small constant $\alpha$.
\end{enumerate}
The algorithm of \cite{cdk23} uses such five strings to construct a string $\tilde{s}$ with good agreement with the optimal string $s^*$. 
Such a string $\tilde{s}$ would be a good candidate for a $(2 - \delta)$-approximate center.
Interestingly, towards finding a string $\tilde{s}$ with good agreement with $s^*$ we observe that for most pairs $(a, b)$ of symbols, the five strings $s_1, ..., s_5$ help find the correct relative ordering of $(a, b)$ in $s^*$. Let us define the set, $Bad$ of bad symbols as:
$$
Bad \coloneqq \cup_{i \neq j \in \{1, 2, 3, 4, 5\}} (I_{s_i} \cap I_{s_j}), 
$$
and the set of good elements as $Good \coloneqq [d]\setminus Bad$.
The reason for categorizing the $d$ symbols into good and bad is that if we consider any symbol $a \in Good$, then it belongs to at most one of the five sets $I_{s_1}, ..., I_{s_5}$. 
This means that for any pair $(a, b)$ of symbols in $Good$ (i.e., $a, b \in Good$), in at least three out of the five strings $s_{1}, ..., s_{5}$, the relative ordering of $a$ and $b$ match that in the optimal string $s^*$. In other words, for any such pair, the five strings reveal their relative ordering in $s^*$. We can use this information to construct a string that has good agreement with $s^*$. 
The main issue is that we do not know which symbols are bad or good. 
We know that the number of bad symbols is much smaller than that of good symbols. This follows from property (B):
$$
|Bad| \leq 4 \veps |I_{s_1}| + 3 \veps |I_{s_2}| + 2 \veps |I_{s_3}| + \veps |I_{s_4}| \leq 10 \veps |I_{s_4}|.
$$
Using $\veps \leq 1/100$, we get that the fraction of bad symbols is not more than $(1/10)^{th}$.
Let us construct a tournament graph $T = ([d], E)$ where the nodes are the $d$ symbols, and there is a directed edge from symbol $a$ to symbol $b$ if and only if $a$ comes before $b$ in at least three out of five strings $s_{1}, ..., s_{5}$. The goal is to use $T$ to find an ordering such that the ordering of ``most" symbols matches that in $s^*$. This is where our algorithm deviates from that of \cite{cdk23}. 
The algorithm in \cite{cdk23} achieves this by finding and removing cycles of length three, thereby removing all bad symbols since every such cycle must have a bad symbol. Once all the bad symbols have been removed, the remaining graph has only good symbols . Since the relative ordering of every pair of good symbols matches that in $s^*$, the correct ordering of the remaining symbols is obtained. The time cost of this algorithm in \cite{cdk23} is $O(d^3)$.\footnote{One can also use the 2-factor approximation algorithm of Lokshtanov et al.~\cite{fvst} for the Feedback Vertex Set problem on Tournament graphs (FVST). However, the running time of their algorithm, $O(d^{12})$, is even worse than cubic.}
We give an algorithm with time linear in $d$ to achieve the same goal using our {\em robust sorting} framework. 
Before continuing with the remaining analysis, it is important to appreciate the fact that an $\tilde{O}(d)$-time algorithm for sorting with an imperfect comparator does not even get to make all $\binom{d}{2}$ pairwise comparisons (and hence is sublinear in the size of the tournament $T$).
When applying our robust sorting framework, total ordering is per the ordering of symbols in the string $s^*$, and the comparison operator {\tt $COMPARE_{s_{1}, ..., s_{5}}(a, b)$} checks if at $a$ is before $b$ in at least three out of five of the strings $s_{1}, ..., s_{5}$.
The remaining aspects, including the good and bad symbols, perfectly align with our current goal. 
That is, the comparison operator gives correct relative ordering (in terms of being consistent with $s^*$) for a pair in which both symbols are good, but if even one symbol in the pair is bad, then no such guarantees exist.
Note that the running time of our algorithm {\tt ROBUST-SORT} is nearly linear in $d$ (ignoring polylogarithmic factors in $d$). 
Using Theorem~\ref{thm:rsort}, we get that that {\tt ROBUST-SORT} returns a string $\tilde{s}$ with the property that $\lcs(s^*, \tilde{s}) \geq d - 4 |Bad|$. 
The rest of the analysis is again similar to \cite{cdk23}. We have:
$$
\Delta(\tilde{s}, s^*) = d - \lcs(\tilde{s}, s^*) \leq 4 |Bad| \stackrel{(B)}{\leq} 40 \veps |I_{i_4}| \stackrel{(C)}{\leq} 40 \veps (1 + \alpha) \frac{OPT}{n}.
$$
Using this, we can now calculate the cost with respect to the center $\tilde{s}$.
$$
Obj(S, \tilde{s}) = \sum_{x \in S} \Delta(x, \tilde{s}) 
\stackrel{(triangle)}{\leq} \sum_{x \in S} \left( \Delta(x, s^*) + \Delta(\tilde{s}, s^*)\right)
\leq \left( 1 + 40 \veps (1+\alpha)\right) \cdot OPT \\
\leq (1.999) \cdot OPT.
$$
The last inequality is by a specific choice of $\veps$ and $\alpha$. We summarize our main result in the next theorem and give the detailed analysis in the following subsections. 

\begin{theorem}[Main theorem for Ulam-$k$-Median]
There is a randomized algorithm for the Ulam $k$-median problem that runs in time $\tilde{O}(nd \cdot k^k)$ and returns a center set $C$ with $Obj(S, C) \leq (2-\delta) \cdot OPT$ with probability at least $0.99$, where $\delta$ is a small but fixed constant.
\end{theorem}

\subsection{A sampling algorithm for Ulam-$1$-Median}
A crucial ingredient in obtaining FPT approximation schemes (i.e., $(1+\veps)$-approximation algorithm) for the $k$-means/median problems is a lemma that shows that for any optimal cluster, a $(1+\veps)$-approximate center for the cluster can be found using a small number (typically $O(\frac{1}{\veps})$) of uniformly sampled points from that cluster. Such a lemma is typically stated in terms of the 1-means/median problem since there is only one cluster. For the $1$-means problem, the lemma is known as Inaba's Lemma~\cite{inaba}. A similar lemma has been shown for the $1$-median problem~\cite{kss}. 
Such a lemma allows us to find good centers as long as we can manage a large enough set of uniformly sampled points from every optimal cluster.
If the optimal clusters are {\em balanced} (i.e., all clusters have a roughly equal number of points), then this task is simple -- {\em uniformly sample and consider all possible partitions}. If the clusters are not balanced, we can use the trick of $D$-sampling that boosts the probability of sampling points from small-sized clusters. This is the core of the argument in the $D$-sampling-based approximation schemes~\cite{jks,bjk18}. 
We give such a sampling lemma for the Ulam-$1$-median problem where the approximation guarantee is $(2-\delta)$ instead of $(1+\veps)$ as in the $k$-means/median problems.
In conjunction with $D$-sampling, this gives an FPT  $(2-\delta/100)$-approximation algorithm for the Ulam-$k$-median problem. Here $\delta$ is a small constant. 
We state this lemma more formally below: 

\begin{lemma}[A sampling lemma]\label{lemma:sampling}
Let $S$ denote any input for the Ulam-$1$-median problem, a subset of permutations of $1, 2, ..., d$. There exists a constant $\eta$ and an algorithm {\tt ULAM1} such that when given as input a set $X \subset S$ consisting of $\eta$ permutations chosen uniformly at random from $S$, {\tt ULAM1} outputs a list of $O(\eta)$ centers such that with probability at least 0.99, at least one center $\tilde{x}$ in the list satisfies $Obj(S, \tilde{x}) \leq (2-\delta) \cdot OPT$, where $\delta$ is a small but fixed constant. Moreover, the running time of {\tt ULAM1} is $\tilde{O}(d)$.
\end{lemma}

In the remaining section, we will build the key ingredients for the algorithm ${\tt ULAM1}$.
We will use a few constants as parameters in our analysis. 
The value of these parameters will be determined at the end of the analysis.
The first parameter is $c_1$, which captures the fraction of points that are close to the optimal center.
%Consider an optimal cluster with points $S = \{x_1, ..., x_n\}$ 
Let $S = \{x_1, ..., x_n\}$
and let $x^{*}$ be an optimal median of $S$. So, $OPT = \sum_{x_i \in S} \Delta(x_i, x^{*})$. 
For each $x_i$, let $I_i$ denote the set of unaligned symbols with $x^{*}$. So, $\Delta(x_i, x^*) = |I_i|$. 
Without loss of generality, we will assume that $|I_1| \leq |I_2| \leq ... \leq |I_n|$. 
In the next lemma, we argue that if a significant number of points are sufficiently close to $x^*$, then a randomly chosen point will be a good center with good probability. 

\begin{lemma}\label{lemma:1}
Let $0.0001 \leq \alpha \leq 1$.
If $|I_j| \leq (1 - \alpha) \cdot \frac{OPT}{n}$ for some $j \geq c_1 \cdot n$, then for a randomly chosen point $x$ from $\{x_1, ..., x_n\}$ satisfies $Obj(S, x) \leq (2 - \alpha) \cdot OPT$ with probability at least $c_1$.
\end{lemma}
\begin{proof}
The point $x$ chosen has index $l \leq j$ with probability at least $c_1$. For such a point $x_l$, we have:
$$
Obj(S, x_l) \leq \sum_{x_i \in S} \Delta(x_i, x_l)
\stackrel{(triangle)}{\leq} \sum_{x_i \in S} (\Delta(x_i, x^*) + \Delta(x_l, x^*)) \leq (2 - \alpha) OPT.
$$
This completes the proof of the lemma.\qed
\end{proof}

For the remaining analysis, let $\alpha \in (0, 1/3]$ be a fixed constant, and we will assume that the least integer $j$ such that $|I_j| > (1-\alpha) \frac{OPT}{n}$ satisfies $j < c_1 \cdot n$. 
This means only a small number of points are sufficiently close to the optimal center $x^*$. 
In this case, we would need to do more than randomly pick a center from $\{x_1, ..., x_n\}$ to find a good center.
Just as we defined {\em near} points in the previous lemma, we will define {\em far} and {\em middle} points for the remaining analysis. Let $s$ be the least integer such that $|I_s| > (1 - \alpha) \frac{OPT}{n}$ and let $t$ be the least integer such that $|I_t| > (1 + \alpha) \frac{OPT}{n}$. Then, we define the following sets:
\begin{itemize}
\item Near points: $Near = \{x_1, ..., x_{s-1}\}$,
\item Middle points: $\bar{F} = \{x_s, ..., x_{t-1}\}$,
\item Far points: $F = \{x_t, ..., x_{n}\}$.
\end{itemize}
We are working under the assumption that 
\begin{equation}\label{eqn:assumption-1}
|Near| < c_1 \cdot n.
\end{equation}
Just as we arrived at the above condition, we also show that $OPT_{F}$ cannot be too large. Otherwise, a random point is a good center with a sufficiently large probability.

\begin{lemma}\label{lemma:2}
Let $c_2 \in (0, 1/3]$ be any constant.
If $OPT_{F} \geq c_2 \cdot OPT$, then a randomly chosen point $x$ from $\{x_1, ..., x_n\}$ satisfies $Obj(S, x) \leq \left(2 - \left(\frac{3\alpha c_2}{8}\right)^2 \right)OPT$ with probability at least $\frac{\alpha}{2 + \alpha} \cdot \left(1 - \frac{1}{1 + \frac{3\alpha c_2}{16}} \right)$.
\end{lemma}
\begin{proof}
Let us define the set $T = \left\{x_j : \Delta(x_j, x^*) \leq (1+\frac{\alpha}{2})\frac{OPT}{n}\right\}$.
From Markov's we know that $|T| \geq \frac{\alpha}{2+\alpha} \cdot n$. 
The cost of a point $x_i \in F$ from any center $x_j \in T$ is given by:
\begin{eqnarray}\label{eqn:lemma2eq1}
\Delta(x_i, x_j) &\leq& \Delta(x_i, x^*) + \Delta(x_j, x^*) \nonumber \\
&\leq& \Delta(x_i, x^*) + \left(1 + \frac{\alpha}{2} \right) \frac{OPT}{n} \nonumber \\
&\leq& \Delta(x_i, x^*) + \frac{(1+\alpha/2)}{1+\alpha} \Delta(x_i, x^*) \nonumber\\
&\leq& \left( 2-\frac{3\alpha}{8}\right)\Delta(x_i, x^*)
\end{eqnarray}
We will calculate the expected cost of $S$ from a center $c$ chosen uniformly at random from $T$.
\begin{eqnarray*}
\E[Obj(S, c)] &=& \frac{1}{|T|} \sum_{x_j \in T} Obj(S, x_j) \\
&=& \frac{1}{|T|} \sum_{x_j \in T}  \left( \sum_{x_i \in S\setminus F} \Delta(x_i, x_j) + \sum_{x_i \in F} \Delta(x_i, x_j)\right)\\
&\leq& \frac{1}{|T|} \sum_{x_j \in T} \sum_{x_i \in S\setminus F} \Delta(x_i, x_j) + \sum_{x_i \in F} \left(2 - \frac{3\alpha}{8}\right)\Delta(x_i, x^*) \quad \textrm{(using \ref{eqn:lemma2eq1})} \\
&\leq& \frac{1}{|T|} \sum_{x_j \in T} \sum_{x_i \in S\setminus F} \left( \Delta(x_i, x^*) + \Delta(x_j, x^*)\right) +  \left(2 - \frac{3\alpha}{8}\right)OPT_F\\
&\leq& 2 \cdot OPT_{S\setminus F} + \left(2 - \frac{3\alpha}{8}\right)OPT_F\\
&=& 2 \cdot OPT - \frac{3\alpha}{8} \cdot OPT_F\\
&\leq& 2 \cdot OPT - \frac{3\alpha c_2}{8} \cdot OPT = \left(2 - \frac{3\alpha c_2}{8} \right) OPT\\
\end{eqnarray*}
So, by Markov, the probability that the cost for a randomly chosen center from $T$, the cost is larger than $\left( 2 - \frac{3\alpha c_2}{8}\right) \cdot \left( 1 + \frac{3 \alpha c_2}{16}\right) OPT = \left( 2 - (3\alpha c_2/8)^2\right) OPT$ is at most $\frac{1}{1 + \frac{3\alpha c_2}{16}}$. 
So, the probability that for a randomly chosen center from $S$, the cost is at most $\left( 2 - (3\alpha c_2/8)^2\right) OPT$ is at least $\frac{\alpha}{2 + \alpha} \cdot \left(1 - \frac{1}{1 + \frac{3\alpha c_2}{16}} \right)$
%$ = \geq \frac{1}{343}$ 

This completes the proof of the lemma. \qed
\end{proof}

The above lemma says that for any constant $c_2 \in (0, 1/3]$, if $OPT_F \geq c_2 \cdot OPT$, then a randomly chosen point breaks the 2-approximation for the cost. 
So, we shall assume that $OPT_F < c_2 \cdot OPT$ for the remaining analysis. Using this with the fact that $OPT = OPT_{Near} + OPT_{\bar{F}} + OPT_F$ and $|Near| \leq c_1 \cdot n$, we get that:
\begin{equation}\label{eqn:1}
OPT_{\bar{F}} \geq (1 - c_2 - c_1(1-\alpha)) \cdot OPT.
\end{equation}
We define the $\veps$-``ball" around $x_j \in \bar{F}$, denoted by $B(x_j, \veps)$, to be the set of all $x_i$'s that satisfy the following two conditions:
\begin{enumerate}
\item $x_i \in \bar{F}$,
\item $|I_i \cap I_j| > \veps |I_j|$.
\end{enumerate}

In other words $B(x_j, \veps) = \{x_i \in \bar{F} : |I_i \cap I_j| > \veps |I_j|\}$.
Let $G$ be the subset of all those points $x_j \in \bar{F}$ such that $OPT_{B(x_j, \veps)} \geq c_3 \cdot OPT_{\bar{F}}$, where $c_3 \leq 1$ is some constant to be decided later.  
First, we obtain a lower bound on the cardinality of $B(x_j, \veps)$ for $x_j \in G$.

\begin{lemma}\label{lemma:3}
For any $x_j \in G$, $|B(x_j, \veps)| \geq \frac{c_3(1 - c_2 - c_1(1-\alpha))}{1+\alpha} \cdot n$.
\end{lemma}
\begin{proof}
Since for every $x_i \in \bar{F}$, $|I_i| \leq (1+\alpha)\cdot \frac{OPT}{n}$, we have:
$$
|B(x_j, \veps)| \cdot (1+\alpha) \cdot \frac{OPT}{n} \geq OPT_{B(x_j, \veps)} \geq c_3 \cdot OPT_{\bar{F}}.
$$
Using Eqn. (\ref{eqn:1}), we obtain the statement of the lemma.\qed
\end{proof}

Next, we analyze the cost of opening a center at $x_j$ in the subset $G$. The cost analysis can be broken into the following cases:
\begin{enumerate}
\item[(i)] For all $x_i \in F$, $\Delta(x_i, x_j) \leq \Delta(x_i, x^*) + \Delta(x_j, x^*) \leq \Delta(x_i, x^*) + \Delta(x_i, x^*) = 2 \cdot \Delta(x_i, x^*)$.
\item[(ii)] For all $x_i \in \bar{F}$, $\Delta(x_i, x_j) \leq \Delta(x_i, x^*) + \Delta(x_j, x^*) \leq \Delta(x_i, x^*) + (1+3\alpha) \cdot \Delta(x_i, x^*) = (2+3\alpha) \cdot \Delta(x_i, x^*)$. The second inequality follows from (i) $\Delta(x_j, x^*) \leq (1+\alpha) (OPT/n)$, and (ii) $\Delta(x_i, x^*) \geq (1-\alpha)(OPT/n)$.
\item[(iii)] For all $x_i \in Near$, $\Delta(x_i, x_j) \leq \Delta(x_i, x^*) + \Delta(x_j, x^*) \leq \Delta(x_i, x^*) + (1+\alpha)\cdot (OPT/n)$.
\item[(iv)] For all $x_i \in B(x_j, \veps)$,  $\Delta(x_i, x_j) \leq |I_i| + |I_j| - |I_i \cap I_j| \leq |I_i| + (1-\veps) |I_j| = \Delta(x_i, x^*) + (1-\veps) \Delta(x_j, x^*)$. The second inequality follows from $x_i \in B(x_j, \veps)$.
\end{enumerate}

The last inequality gives us the following lemma.

\begin{lemma}\label{lemma:4}
Let $x_j \in G$. Then $\sum_{x_i \in B(x_j, \veps)} \Delta(x_i, x_j) \leq \left(2 - \veps (1+3\alpha-3\alpha/\veps)\right) \cdot OPT_{B(x_j, \veps)}$
\end{lemma}
\begin{proof}
Let $R = B(x_j, \veps)$. The lemma follows from the sequence of inequalities below:
\begin{eqnarray*}
\sum_{x_i \in R} \Delta(x_i, x_j) &\leq& \sum_{x_i \in R} \left( \Delta(x_i, x^*) + (1-\veps) \Delta(x_j, x^*)\right) \quad \textrm{(using (iv))}\\
&=& OPT_{R} + (1-\veps) \cdot \left( \sum_{x_i \in R, |I_i| \leq |I_j|} |I_j| + \sum_{x_i \in R, |I_i| > |I_j|} |I_i|\right)\\
&\leq& OPT_{R} + (1-\veps) \cdot \left( \sum_{x_i \in R, |I_i| \leq |I_j|} \frac{1+\alpha}{1-\alpha}|I_i| + \sum_{x_i \in R, |I_i| > |I_j|} |I_i|\right) \\
&& \quad \textrm{(since $\forall x_i\in R, (1-\alpha) (OPT/n) \leq |I_i| \leq (1+\alpha) (OPT/n)$ )}\\
&\leq& OPT_R + (1-\veps) \cdot (1+3\alpha) \cdot OPT_R\\
&\leq& OPT_R \cdot \left(2 - \veps (1+3\alpha-3\alpha/\veps)\right).
\end{eqnarray*}
This completes the proof of the lemma.\qed
\end{proof}
We can now analyze the cost with respect to a center $x_j$ chosen from the subset $G$.

\begin{lemma}\label{lemma:5}
Let $x_j \in G$ and  $\left( c_3 \veps (1+3\alpha) - 3\alpha\right) (1 - c_2 - c_1(1-\alpha)) - c_1(1+\alpha) \geq 0.0001$.
Then $Obj(S, x_j) \leq (1.999) \cdot OPT$.
%, where $c_4 = \left( c_3 \veps (1+3\alpha) - 3\alpha\right) (1 - c_2 - c_1(1-\alpha)) - c_1(1+\alpha)$.
\end{lemma}
\begin{proof}
Let $R = B(x_j, \veps)$. The bound follows from the sequence of inequalities below:
\begin{eqnarray*}
Obj(S, x_j) &=& \sum_{x_i \in Near} \Delta(x_i, x_j)  + \sum_{x_i \in \bar{F}\setminus R} \Delta(x_i, x_j) + \sum_{x_i \in F} \Delta(x_i, x_j) + \sum_{x_i \in R} \Delta(x_i, x_j)\\
&\leq& \sum_{x_i \in Near} \Delta(x_i, x^*) + c_1(1+\alpha)OPT + \sum_{x_i \in \bar{F}\setminus R} \Delta(x_i, x_j) + \sum_{x_i \in F} \Delta(x_i, x_j) + \sum_{x_i \in R} \Delta(x_i, x_j)\\
&& \quad \textrm{(Using (iii) of cost analysis and $|Near| < c_1 n$)}\\
&\leq&\sum_{x_i \in Near} \Delta(x_i, x^*) + c_1(1+\alpha)OPT + (2+3\alpha) \sum_{x_i \in \bar{F}\setminus R} \Delta(x_i, x^*) + \sum_{x_i \in F} \Delta(x_i, x_j) + \sum_{x_i \in R} \Delta(x_i, x_j)\\
&& \quad \textrm{(Using (ii) of cost analysis)}\\
&\leq&\sum_{x_i \in Near} \Delta(x_i, x^*) + c_1(1+\alpha)OPT + (2+3\alpha) \sum_{x_i \in \bar{F}\setminus R} \Delta(x_i, x^*) + 2 \sum_{x_i \in F} \Delta(x_i, x^*) + \sum_{x_i \in R} \Delta(x_i, x_j)\\
&& \quad \textrm{(Using (i) of cost analysis)}\\
&=& OPT_{Near} + (2+3\alpha) OPT_{\bar{F}\setminus R} + 2 OPT_{F} + c_1(1+\alpha)OPT + \sum_{x_i \in R} \Delta(x_i, x_j)\\
&\leq& OPT_{Near} + (2+3\alpha) OPT_{\bar{F}\setminus R} + 2 OPT_{F} + c_1(1+\alpha)OPT + \left(2 - \veps (1+3\alpha-3\alpha/\veps)\right) \cdot OPT_{R}\\
&\leq& 2 OPT - \left( c_3 \veps (1+3\alpha) - 3\alpha\right) OPT_{\bar{F}} + c_1(1+\alpha)OPT \qquad \textrm{(Using $OPT_R \geq c_3 OPT_{\bar{F}}$)}\\
&\leq& 2 OPT -  \left( c_3 \veps (1+3\alpha) - 3\alpha\right) (1 - c_2 - c_1(1-\alpha))OPT + c_1(1+\alpha)OPT \qquad \textrm{(Using Eqn. (\ref{eqn:1}))}
\end{eqnarray*}
This completes the proof of the lemma.\qed
\end{proof}

For the analysis of the algorithm, we consider two cases based on the size of the subset $G$:
\begin{itemize}
\item \underline{Case 1 ($|G| \geq c_4 n$)}: In this case, by Lemma~\ref{lemma:5}, a randomly chosen center will give a $1.999$-approximation with probability at least $c_4$.
\item \underline{Case 2 ($|G| < c_4 n$)}: This case is analyzed in the remainder of the analysis. Here, we will argue that with high probability, five uniformly sampled points give a center that breaks the 2-approximation.
\end{itemize}
Consider any point $x_j \in \bar{F}\setminus G$. We know that $OPT_{B(x_j, \veps)} < c_3 OPT_{\bar{F}}$. Let us use this to upper bound the size of $B(x_j, \veps)$.
Using $|B(x_j, \veps)| \cdot (1-\alpha) \cdot \frac{OPT}{n} \leq OPT_{B(x_j, \veps)} \leq c_3 \cdot OPT_{\bar{F}} \leq c_3 \cdot OPT$, we get
$$
|B_{x_j, \veps}| \leq \frac{c_3}{(1-\alpha)} \cdot n
$$
On the other hand, we can use $|\bar{F}| \cdot (1+\alpha) \cdot \frac{OPT}{n} \geq OPT_{\bar{F}} \geq (1 - c_2 - c_1(1-\alpha)) \cdot OPT$, we get $|\bar{F}| \geq \frac{(1 - c_2 - c_1(1-\alpha))}{1+\alpha} \cdot n$, which gives
$$
|\bar{F} \setminus G| \geq \left(\frac{(1 - c_2 - c_1(1-\alpha))}{1+\alpha} - c_4 \right) \cdot n
$$
The following lemma gives the condition under which five randomly chosen points help find a center that beats the 2-approximation.

\begin{lemma}\label{lemma:6}
Let $x_{i_1}, x_{i_2}, ..., x_{i_5}$ be five points chosen at random from the dataset. The following three properties are satisfied with a probability of at least $\frac{1}{20^5}$:
\begin{enumerate}
\item $x_{i_1}, ..., x_{i_5} \in \bar{F} \setminus G$,
\item $|I_{i_1}| \leq |I_{i_2}| \leq |I_{i_3}| \leq |I_{i_4}| \leq |I_{i_5}|$, and
\item For every pair $r > s \in \{1, ..., 5\}$, $|I_{i_r} \cap I_{i_s}| \leq \veps |I_{i_s}|$.
\end{enumerate}
\end{lemma}
\begin{proof}
We will work with the assumption that $\left(\frac{(1 - c_2 - c_1(1-\alpha))}{1+\alpha} - c_4 \right) \geq \frac{1}{2}$, which implies that $|\bar{F}\setminus G| \geq \frac{n}{2}$.
Further, we will work with the assumption that $\frac{c_3}{1-\alpha} \leq \frac{1}{10}$, which implies $|B_{x_{j}, \veps}| \leq \frac{c_3 n}{1-\alpha} \leq \frac{n}{10} \leq \frac{|\bar{F}\setminus G|}{5}$. Let us consider the elements of $\bar{F} \setminus G$ in increasing order of distance from $x^*$ and, in particular, the first block of $\frac{|\bar{F} \setminus G|}{10}$ elements. The probability that the randomly chosen $x_{i_1}$ belongs to this block is at least $\frac{1}{2} \cdot \frac{1}{10} = \frac{1}{20}$. 
We will now recurse on the set of elements $\bar{F} \setminus (G \cup B_{x_{i_1}, \veps})$ when selecting $x_{i_2}, ..., x_{i_5}$. So, the probability that the five randomly chosen points satisfy the three conditions is at least $(1/20)^5$. \qed
\end{proof}

We will use the five points $x_{i_1}, ..., x_{i_5}$ of the previous lemma to construct a string $x$ such that $\Delta(S, x) \leq 1.999 \cdot OPT$. 
The string $x$ should have agreement on the ordering of the symbols when compared to the optimal string $x^*$. Interestingly, for most pairs $(a, b)$ of symbols, the five strings $x_{i_1}, ..., x_{i_5}$ help find the correct relative ordering of $(a, b)$ in $x^*$. We follow the same high-level outline as in \cite{cdk23}.
Let us define the set $Bad$ of bad symbols as:
$$
Bad \coloneqq \cup_{r \neq s \in \{1, 2, 3, 4, 5\}}{(I_{i_r} \cap I_{i_s})},
$$
and the set $Good$ of good symbols as $Good = [d] \setminus B$. The reason for categorizing the $d$ symbols into good and bad is that if we consider any symbol $a \in Good$, then it belongs to at most one of the five sets $I_{i_1}, ..., I_{i_5}$. This means that for any pair $(a, b)$ of symbols in $Good$ (i.e., $a, b \in Good$), in at least three out of the five strings $x_{i_1}, ..., x_{i_5}$, the relative ordering of $a$ and $b$ match that in the optimal string $x^*$. In other words, for any such pair, the five strings reveal their relative ordering in $x^*$. We can use this information to construct a string that has good agreement with $x^*$. 
The main issue is that we do not know which symbols are bad or good. 
We know that the number of bad symbols is much smaller than that of good symbols. This follows from:
$$
|Bad| \leq 4 \veps |I_{i_1}| + 3 \veps |I_{i_2}| + 2 \veps |I_{i_3}| + \veps |I_{i_4}| \leq 10 \veps |I_{i_4}|.
$$
Using $\veps \leq 0.01$, we get that the fraction of bad symbols is not more than $(1/10)^{th}$.
Let us construct a tournament graph $T = ([d], E)$ where the nodes are the $d$ symbols, and there is a directed edge from symbol $a$ to symbol $b$ iff $a$ comes before $b$ in at least three out of five strings $x_{i_1}, ..., x_{i_5}$. The goal is to use $T$ to find an ordering such that the ordering of ``most" symbols matches that in $x^*$. This is where our algorithm deviates from that of \cite{cdk23}. 
The algorithm in \cite{cdk23} achieves this by finding and removing cycles of length three, thereby removing all bad symbols since every such cycle must have a bad symbol. The remaining graph has only good symbols once all the bad symbols have been removed. Since the relative ordering of every pair of good symbols matches that in $x^*$, the correct ordering of the remaining symbols is obtained. The time cost of this algorithm in \cite{cdk23} is $O(d^3)$, which makes the algorithm not linear. We give a much more efficient algorithm to achieve the same goal using our {\em robust sorting} framework. 
When applying our robust sorting framework, total ordering is per the ordering of symbols in the string $x^*$, and the comparison operator {\tt $COMPARE_{x_{i_1}, ..., x_{i_5}}(a, b)$} checks if at $a$ is before $b$ in at least three out of five of the strings $x_{i_1}, ..., x_{i_5}$.
The remaining aspects, including the good and bad symbols, perfectly align with our current goal.
Note that the running time of our algorithm {\tt ROBUST-SORT} is nearly linear in $d$ (ignoring polylogarithmic factors in $d$). 
Using Theorem~\ref{thm:rsort}, we get that that {\tt ROBUST-SORT} returns a string $x_T$ with the property that $LCS(\tilde{x}, x^*) \geq d - 4|Bad|$.
%the ordering of symbols in $x_T$ matches $x^*$. 
%Let $\tilde{x}$ be the string obtained by concatenating the remaining symbols (i.e., symbols not in the string $x_T$) in an arbitrary order with $x_T$. Then we have:
%$$
%LCS(\tilde{x}, x^*) \geq d - 4|Bad|
%$$
The rest of the analysis is similar to \cite{cdk23}. We have:
$$
\Delta(\tilde{x}, x^*) = d - LCS(\tilde{x}, x^*) \leq 4 |Bad| \leq 40 \veps |I_{i_4}| \leq 40 \veps (1 + \alpha) \frac{OPT}{n} \quad \textrm{(since $x_{i_4} \in \bar{F}$)}
$$
Using this, we can now calculate the cost with respect to the center $\tilde{x}$.
\begin{eqnarray*}
Obj(S, \tilde{x}) &=& \sum_{x_i \in S} \Delta(x_i, \tilde{x}) \\
&\leq& \sum_{x_i \in S} \left( \Delta(x_i, x^*) + \Delta(\tilde{x}, x^*)\right) \qquad \textrm{(by triangle inequality)}\\
&\leq& \left( 1 + 40 \veps (1+\alpha)\right) \cdot OPT \\
&\leq& (1.999) \cdot OPT.
\end{eqnarray*}
The last inequality is by our choice of $\veps$ and $\alpha$. 
In fact, we will fix all the constants introduced in the analysis and ensure that the constraints posed on them are satisfied. Let's write the constraints on all the constants.
\begin{eqnarray*}
0.0001 \leq \alpha &\leq& \frac{1}{3} \\
0 < c_2 &\leq& \frac{1}{3} \\
 2 - \left(\frac{3\alpha c_2}{8}\right)^2 &\leq& 2 - 10^{-10}\\
\frac{\alpha}{2 + \alpha} \cdot \left( 1 - \frac{1}{1 + \frac{3\alpha c_2}{16}}\right) &\geq& 10^{-12}\\
 \left( c_3 \veps (1+3\alpha) - 3\alpha\right) (1 - c_2 - c_1(1-\alpha)) - c_1(1+\alpha) &\geq& 0.0001\\ 
\left(\frac{(1 - c_2 - c_1(1-\alpha))}{1+\alpha} - c_4 \right) &\geq& \frac{1}{2}\\
\frac{c_3}{1 - \alpha} &\leq& \frac{1}{10} \\
40 \veps (1+\alpha) &\leq& 0.999 \\
\end{eqnarray*}
The following values of constants satisfy the above constraints\footnote{We fixed constants that break the approximation barrier of 2. We did not try to optimize the constants beyond breaking the 2-barrier.}:
$$
\veps = \frac{1}{41}, c_3 = \frac{1}{11}, \alpha = \frac{c_3 \veps}{6}, c_1 = c_2 = c_4 = \frac{c_3 \veps}{10}.
$$
We now have all the ingredients to prove Lemma~\ref{lemma:sampling}. Consider eight randomly sampled points $y_1, ..., y_8$ from $S$ and consider the list of centers $\{y_1, y_2, y_3,$ {\tt ROBUST-SORT(\tt {$[d]$, $COMPARE_{y_4, ..., y_8}(., .)$})}\}. If the condition of Lemma~\ref{lemma:1} holds, then $y_1$ is a good center with probability at least $c_1$. If the condition of Lemma~\ref{lemma:2} holds, then $y_2$ is a good center with probability at least $10^{-12}$. If $|G| \geq c_4 n$, then by Lemma~\ref{lemma:6}, $y_3$ is a good center with probability at least $c_4$. In the remaining case, the five centers $y_4, ..., y_8$ give a good center (using robust-sorting) with a probability of at least $\frac{1}{20^5}$.
So, the probability of success of this randomized strategy is given by: 
$$
\min{\left(c_1, 10^{-12}, c_4, \frac{1}{20^5} \right)} = 10^{-12}.
$$
Let $\eta =16 \cdot 10^{12}$. If we repeat this $O(\eta)$ times and, in the process, sample $O(\eta)$ points, then the probability that our list contains a good center is at least $0.99$. This completes the proof of Lemma~\ref{lemma:sampling}.

\subsection{The Ulam-$k$-Median problem}
In this section, we use the sampling algorithm {\tt ULAM1} for the Ulam-1-median problem defined in Lemma~\ref{lemma:sampling} as a subroutine to design an algorithm for the Ulam-$k$-median problem that gives an approximation guarantee of $(2-\delta)$. 
This algorithm is based on the $D$-sampling technique.
The algorithm and its analysis follow the $D$-sampling based $(1+\veps)$-approximation results of \cite{jks} for the $k$-means problem.
For readers who are familiar with the $D$-sampling-based algorithms of \cite{jks}, the contents in this section can be summarised in the following manner: {\it the high-level algorithm description and its analysis are the same as the algorithm for the $k$-means problem in \cite{jks}, except that here we use our sampling Lemma~\ref{lemma:sampling}, instead of the Inaba's sampling lemma~\cite{inaba} that is used for the $k$-means problem.}
We provide the details for the unfamiliar reader.

The high-level ideas have already been outlined in the beginning of the previous subsection. Let $S_1, ..., S_k$ be the dataset partition that denotes the optimal $k$ clustering, and let $x^*_1, ..., x^*_k$ denote the optimal centers, respectively. Let us try to use {\tt ULAM1} to find good centers for each of $S_1, ..., S_k$. We would need $\eta$ uniformly sampled points each from $S_1, ..., S_k$. 
The issue is that the optimal clustering $S_1, ..., S_k$ is not known.
If the clusters were balanced, i.e., $|S_1| \approx |S_2| \approx ... \approx |S_k|$, then uniformly sampling $O(\eta k)$ points from $S$ and then considering all possible partitions of these points would give the required uniform samples $X_1, ..., X_k$ from each of the optimal clusters. We can then use {\tt ULAM1($X_i$)} to find good center candidates for $S_i$. 
In the general case, where the clusters may not be balanced, we use the $D$-sampling technique to boost the probability of sampling from small-sized optimal clusters, which may get ignored when sampling uniformly at random from $S$. We sample recursively. In the first step, we sample uniformly at random from $S$. This sampled set has many representatives from the largest optimal set $S_{max}$. So. if we sample $O(\eta k)$ points and consider every possible subset of size $\eta$, then at least one represents a uniform sample from $S_{max}$ with high probability. We will try out every possible subset, and for each subset, we will consider the list of possible centers produced by {\tt ULAM1}. 
For every choice of the first center $c$, we will recursively sample, but sample with $D$-distribution with respect to the center set $C \coloneqq \{c\}$.\footnote{Sampling with $D$-distribution or in other words, $D$-sampling with respect to a center set $C$, means sampling from a distribution over the dataset where the sampling probability of a point is proportional to its distance from the closest center in the set $C$.}
This way, we will add $k$ centers to our center set $C$ in a recursive algorithm of depth $k$. 
We can label the nodes in the corresponding recursion tree with the centers accumulated until that point in the algorithm. The root node starts with an empty set $C \coloneqq \{\}$ of centers, while every node at depth $i$ has a center set of size $i$.
We will argue that with high probability, there is one path in the recursion tree where the leaf node has good centers $c_1, ..., c_k$ (for $(2-\delta)$ approximation) for every cluster $S_1, ..., S_k$. 
Algorithm~\ref{fig:k} gives a pseudocode description of this algorithm. In subsequent discussions, we will use the the distance function $\Delta$ to also to write the objective function, i.e., for set of points $X$ and center set $C$, $\Delta(X, C) \coloneqq Obj(X, C) = \sum_{x \in X} \min_{c \in C} \Delta(x, c)$.

\begin{center}
\begin{Algorithm}[h]
\begin{boxedminipage}{5.8in}
{\tt ULAMk($S, k$)}

\hspace{0.1in} - Let $\veps = \delta/100$ and $N = \frac{16 \eta k}{\veps^2}$.

\hspace{0.1in} - Repeat $O(2^k)$ times and output the set of centers $C$ that give least cost

\hspace{0.3in} - Make a call to {\tt Sample-centers$(S, k, \veps, 0, \{\})$} and select $C$ from the set  

\hspace{0.4in} of solutions that give the least cost.

\vspace{0.1in}

{\tt Sample-centers$(S, k, \veps, i, C)$}

\hspace{0.1in} (1) If $(i = k)$ then add $C$ to the set of solutions

\hspace{0.1in} (2) else

\hspace{0.3in} (a) Sample a multiset $T$ of $N$ points with $D$-sampling (w.r.t. centers $C$)

\hspace{0.3in} (b) For all subsets $R \subset T$ of size $\eta$:

\hspace{0.7in} For every center $c$ in the set {\tt ULAM1($R$)}

\hspace{0.9in} (i) $C \leftarrow C \cup \{c\}$.

\hspace{0.9in} (ii) {\tt Sample-centers$(S, k, \veps, i+1, C)$}
\end{boxedminipage}
\caption{$D$-sampling based algorithm for the Ulam $k$-median problem.}
\label{fig:k}
\end{Algorithm}
\end{center}
We show the following main result regarding the algorithm {\tt ULAMk} above. Let $\delta$ be the constant associated with {\tt ULAM1}, i.e., the algorithm gives $(2-\delta)$-approximation for the Ulam 1-median problem.

\begin{theorem}[Main theorem for Ulam-$k$-Median]\label{thm:main-ulamk}
The algorithm {\tt ULAMk} runs in time $\tilde{O}(nd \cdot k^k)$ and returns a center set $C$ with $\Delta(S, C) \leq (2-\delta'') \cdot OPT$ with probability at least $0.99$, where $\delta'' = \delta/100$.
\end{theorem}
\underline{\it Running time of {\tt ULAMk}}: Calls to {\tt Sample-centers} is made $O(2^k)$ times (for probability amplification) within {\tt ULAMk}. The subroutine {\tt Sample-centers} is a recursive algorithm with depth $k$ and branching factor of $\binom{N}{\eta} \cdot O(\eta) = O(k)$.
Calculating the distance of a point from a center costs $O(d \log{d})$ time.
This is because the LCS of two permutations can be found using a reduction to the {\em Longest Increasing Subsequence} problem that takes $O(d \log{d})$ time.
Given this, updating the distribution (for $D$-sampling) after picking every new center costs $\tilde{O}(nd)$ time.
Within every recursive call to {\tt Sample-centers} we also need to make a call to {\tt ULAM1} that costs $\tilde{O}(d)$ time.
So, the total time spent within the calls made to {\tt Sample-centers} is $\tilde{O}(nd) \cdot O(2^k) \cdot O\left( k^k\right) = \tilde{O}(nd \cdot k^k)$.
Finding the best center set out of all possibilities also costs $\tilde{O}(nd \cdot k^k)$ time.
So, the overall running time is $\tilde{O}(nd \cdot k^k)$.

To prove the approximation guarantee in the above theorem, we will show that the recursion tree for {\tt Sample-centers} satisfies the following invariant $P(i)$ for every $i$:
\begin{quote}
$P(i)$: With probability at least $\frac{1}{2^i}$, there exists a node at depth $i$ of the recursion tree labeled with the center set $C^{(i)} \coloneqq \{c_1, ..., c_i\}$ such that the following is satisfied for some $i$ optimal clusters $S_{j_1}, ..., S_{j_i}$:
\begin{equation}\label{eqn:ulamk1}
\Delta(S_{j_1}\cup ... \cup S_{j_i}, C^{(i)}) \leq (2 - \delta') \left(\sum_{r=1}^{i} \Delta(S_{j_r}, x^*_{j_r}) \right), \textrm{where $\delta' = \delta/10$} .
\end{equation}
\end{quote}
Note that if the invariant holds, then {\tt ULAMk} outputs a $(2 - \delta'')$-approximate solution with probability 0.99 since it repeats {\tt Sample-centers} $O(2^k)$ times and picks the best solution. Since $P(0)$ holds trivially, let us show that $P(i+1)$ holds given than $P(0), ..., P(i)$ holds. Consider any node at depth $i$ labeled with centers $C^{(i)}$ satisfying (\ref{eqn:ulamk1}). 
Let $J = \{j_{1}, ..., j_{i}\}$ denote the subset of cluster indices, which, in some sense, are covered by centers in $C^{(i)}$.
Consider the recursive call {\tt Sample-centers($S, k, i+1, C^{(i)}$)}. We will show that with probability at least $1/2$, one of the candidate centers $c$ considered within this recursive call is such that
$$
\Delta(S_{j_1}\cup ... \cup S_{j_i} \cup S_{j_{i+1}}, C^{(i)} \cup \{c\}) \leq (2-\delta') \left(\sum_{r=1}^{i+1} \Delta(S_{j_r}, x^*_{j_r}) \right).
$$
In other words, a previously uncovered cluster $S_{j_{i+1}}$ is now covered (with respect to $(2-\delta')$ approximation).
Within the body of the recursive call {\tt Sample-centers($S, k, i+1, C^{(i)}$)}, a multiset $R$ is selected using $D$-sampling with respect to $C^{(i)}$. 
If the probability of sampling points in $\cup_{l \notin J} S_{l}$ is very small, we argue that $C^{(i)}$ is a good solution for the entire set of points $S$.

\begin{lemma}
Let $0 < \veps \leq \delta'/10$. If $\frac{\sum_{l \notin J} \Delta(S_l, C^{(i)})}{\sum_{l \in [k]} \Delta(S_l, C^{(i)})} \leq \frac{\veps}{2}$, then $\Delta(S, C^{(i)}) \leq (2-\delta'') \cdot OPT$.
\end{lemma}
\begin{proof}
We have:
\begin{eqnarray*}
\Delta(S, C^{(i)}) &=& \sum_{l \in J} \Delta(S_l, C^{(i)}) + \sum_{l \notin J} \Delta(S_l, C^{(i)}) \\
&\leq& \sum_{l \in J} \Delta(S_l, C^{(i)}) + \frac{\veps/2}{1 - \veps/2} \cdot \sum_{l \in J} \Delta(S_l, C^{(i)})\\
&=& \frac{1}{1 - \veps/2} \cdot \sum_{l \in J} \Delta(S_l, C^{(i)}) \\
&=& \frac{1}{1 - \veps/2} \cdot \sum_{l \in [k]} \Delta(S_l, C^{(i)}) \\
&\leq& (2-\delta') \cdot (1+\veps) \cdot OPT\\
&\leq& (2 - \delta'') \cdot OPT.
\end{eqnarray*}
This completes the proof of the lemma. \qed
\end{proof}
Let $j_{i+1}$ be the index of the cluster with the largest value of $\Delta(S_{j_{i+1}}, C^{(i)})$.
The following corollary of the above lemma will be useful for further analysis.

\begin{corollary}
Let $0 < \veps \leq \delta'/10$. If $\Delta(S, C^{(i)}) > (2-\delta'') \cdot OPT$, then 
$\frac{\sum_{l \notin J} \Delta(S_l, C^{(i)})}{\sum_{l \in [k]} \Delta(S_l, C^{(i)})} \geq \frac{\veps}{2k}$.
\end{corollary}
So, there is a lower bound on the probability of sampling an element from $S_{j_{i+1}}$ when $D$-sampling with respect to $C^{(i)}$. However, the probability of sampling every element from $S_{j_{i+1}}$ is not the same. Some elements will have a very small probability of getting sampled. We show that such elements can be assigned to one of the centers in $C^{(i)}$ without a significant cost penalty. 
Let $m_{j_{i+1}} = |S_{j_{i+1}}|$ and $r_{j_{i+1}} = \frac{\Delta(S_{j_{i+1}}, x^*_{j_{i+1}})}{m_{j_{i+1}}}$.
Let $L_{j_{i+1}} \subset S_{j_{i+1}}$ be those points such that: 
$$
\forall p \in L_{j_{i+1}}, \quad \frac{\Delta(p, C^{(i)})}{\Delta(S_{j_{i+1}}, C^{(i)})} \leq \frac{\veps/2}{m_{j_{i+1}}},
$$
and $H_{j_{i+1}} = S_{j_{i+1}} \setminus L_{j_{i+1}}$.

\begin{lemma}
For any point $p \in L_{j_{i+1}}$, $\Delta(p, C^{(i)}) \leq \frac{\veps/2}{1 - \veps/2} \cdot (r_{j_{i+1}} + \Delta(p, x^*_{j_{i+1}}))$.
\end{lemma}
\begin{proof}
Let $c \coloneqq \arg\min_{t \in C^{(i)}}{\Delta(p, t)}$, i.e., the closest center to $p$ in $C^{(i)}$. We have:
\begin{eqnarray*}
\frac{\veps/2}{m_{j_{i+1}}} &\geq& \frac{\Delta(p, C^{(i)})}{\Delta(S_{j_{i+1}}, C^{(i)})} \\
&\geq& \frac{\Delta(p, c)}{\Delta(S_{j_{i+1}}, c)}\\
&\geq& \frac{\Delta(p, c)}{\Delta(S_{j_{i+1}}, x^*{_{j_{i+1}}}) + m_{j_{i+1}} \cdot \Delta(x^*_{j_{i+1}}, c)} \quad \textrm{(using triangle inequality)}\\
&=& \frac{\Delta(p, c)}{m_{j_{i+1}} \cdot r_{j_{i+1}} + m_{j_{i+1}} \cdot \Delta(x^*_{j_{i+1}}, c)} \\
&\geq& \frac{\Delta(p, c)}{m_{j_{i+1}} \cdot r_{j_{i+1}} + m_{j_{i+1}} \cdot \Delta(x^*_{j_{i+1}}, p) + m_{j_{i+1}} \cdot \Delta(p, c)} \ \ \ \textrm{(using triangle inequality)}
\end{eqnarray*}
Simplifying, we get the following: 
\begin{eqnarray*}
\Delta(p, c) &\leq& \frac{\veps/2}{1 - \veps/2} \cdot (r_{j_{i+1}} + \Delta(p, x^*_{j_{i+1}}))
\end{eqnarray*}
This completes the proof of the lemma.\qed
\end{proof}
The next lemma shows that there is only a small cost penalty of assigning points in $L_{j_{i+1}}$ to centers in $C^{(i)}$.

\begin{lemma}
$\Delta(L_{j_{i+1}}, C^{(i)}) \leq \frac{\veps}{2} \cdot \Delta(S_{j_{i+1}}, x^*_{j_{i+1}})$.
\end{lemma}
\begin{proof}
Using the previous lemma, we get: 
\begin{eqnarray*}
\Delta(L_{j_{i+1}}, C^{(i)}) &=& \sum_{p \in L_{j_{i+1}}} \Delta(p, C^{(i)})\\
&\leq&  \frac{\veps/2}{1-\veps/2} \cdot \sum_{p \in L_{j_{i+1}}}  (r_{j_{i+1}} + \Delta(p, x^*_{j_{i+1}}))\\
&\leq&  \frac{\veps/2}{1-\veps/2} \cdot \sum_{p \in S_{j_{i+1}}}  (r_{j_{i+1}} + \Delta(p, x^*_{j_{i+1}}))\\
&=&  \frac{\veps/2}{1-\veps/2} \cdot \sum_{p \in S_{j_{i+1}}} (m_{j_{i+1}} \cdot r_{j_{i+1}} + \Delta(S_{j_{i+1}}, x^*_{j_{i+1}})) \\
&=& \frac{\veps}{1 - \veps/2} \cdot \Delta(S_{j_{i+1}}, x^*_{j_{i+1}})\\
&\leq& \frac{\veps}{2} \cdot \Delta(S_{j_{i+1}}, x^*_{j_{i+1}})
\end{eqnarray*}
This completes the proof of the lemma.\qed
\end{proof}
Suppose in the current recursive call; we are able to find a center $c'$ such that $\Delta(H_{j_{i+1}}, c') \leq (2-\delta) \cdot \Delta(S_{j_{i+1}}, x^*_{j_{i+1}})$, then using the previous lemma, we get: 
\begin{eqnarray*}
\Delta(S_{j_{i+1}}, C^{(i)}\cup \{c'\}) &\leq& \Delta(H_{j_{i+1}}, c') + \Delta(L_{j_{i+1}}, C^{(i)}) \\
&\leq& (2-\delta) \Delta(S_{j_{i+1}}, x^*_{j_{i+1}}) + \frac{\veps}{2} \Delta(S_{j_{i+1}}, x^*_{j_{i+1}}) \\
&\leq& (2-\delta') \cdot \Delta(S_{j_{i+1}}, x^*_{j_{i+1}}).
\end{eqnarray*}
The center $c'$ can be obtained by oversampling using $D$-distribution, which ensures a sufficient number of uniform samples from $H_{j_{i+1}}$ (since each element in $H_{j_{i+1}}$ has some minimal probability of getting sampled) required by the {\tt ULAM1} algorithm.
The probability of $D$-sampling an element $p \in H_{j_{i+1}}$ is at least $\frac{\veps}{2k} \cdot \frac{\veps}{2 m_{j_{i+1}}} = \frac{(\veps^2/4k)}{m_{j_{i+1}}}$. 
So, if we $D$-sample $\frac{16k \eta}{\veps^2}$ points and consider all possible subsets of size $\eta$, then at least one of the subsets will represent a uniform sample from $H_{j_{i+1}}$ with probability at least $1/2$. This fact is shown in the Appendix (see Section~\ref{sec:sampling-lemma}).
This subset, when given as input to {\tt ULAM1}, outputs a candidate center $c'$ such that $\Delta(H_{j_{i+1}}, c') \leq (2-\delta) \cdot \Delta(H_{j_{i+1}}, x^*_{j_{i+1}})$. This completes the proof of Theorem~\ref{thm:main-ulamk}.

\section{A Sampling Theorem}\label{sec:sampling-lemma}

The issue that we address in this section is obtaining uniform samples from a subset $H \subseteq S$ using a sampler that samples from a distribution $\D$ over $S$ that satisfies the following property:
\begin{itemize}
\item For any element $p \in H$, the probability that $p$ is sampled is at least $\frac{\veps}{m}$, where $m = |H|$.
\end{itemize}
The issue in obtaining uniform samples from $H$ is that we cannot check whether a given element belongs to $H$. 
We need $\eta$ uniform samples from $H$.
Our strategy is to sample a multi-set $T$ with $\frac{4\eta}{\veps}$ elements from $S$ using our $\D$-sampler and then consider all possible subsets of size $\eta$. We will argue that with high probability, there exists a subset $T$ of size $\eta$ that is a uniform sample from $H$. 
Let $\D_{\eta}$ denote the probability distribution over multisets of $S$ of size $\eta$ chosen uniformly at random (with replacement) from $H$. Then we want to show that there exists a subset $R \subset T$ of size $\eta$ that has the same distribution as $\D_{\eta}$.
This argument is summarised in the next two lemmas.

\begin{lemma}
\label{lem:clinaba}
Let $Q$ be a set of $n$ points, and $\veps$ be a parameter, $0 < \veps < 1$. Define a random variable $X$ as follows :
with probability $\veps$, it picks an element of $Q$ uniformly at random, and with probability $1-\veps$, it does not
pick any element (i.e., is null). Let $X_1, \ldots, X_\ell$ be $\ell$ independent copies of $X$, where $\ell = \frac{4\eta}{\veps}.$
Let $P$ denote the multi-set of elements of $Q$ picked by $X_1, \ldots, X_\ell$. Then, with probability at least $3/4$,
$P$ contains a subset $U$ of size $\eta$ that has the same distribution as $\D_{\eta}$.
\end{lemma}
\begin{proof}
Define a random variable $I$, which is a subset of the index set $\{1, \ldots, \ell\}$, as follows $I = \{ t : X_t \mbox{ picks an element of $Q$, i.e., it is not null} \}$. 
Conditioned on $I = \{t_1, \ldots, t_r\}$, note that the random variables $X_{t_1}, \ldots, X_{t_r}$ are independent uniform samples from $Q$. 
Thus if $|I| \geq \eta$, then there is a $U \subset P$ that has the same distribution as $\D_{\eta}$. 
But the expected value of $|I|$ is $4 \eta$, and so, from Chernoff bound $|I| \geq \eta$ with probability at least $0.99$. Hence, the statement in the lemma holds.\qed
\end{proof}

\begin{lemma}
\label{lem:final}
With probability at least $3/4$, there exists a subset $R$ of $T$ of size $\eta$ such that $R$ has the same distribution as $\D_{\eta}$
\end{lemma}
\begin{proof}
Let $N = \frac{4\eta}{\veps}$.
Let $Y_1, \ldots, Y_N$ be  $N$ independent random variables defined as follows : 
for any $t$, $1 \leq t \leq N$, $Y_t$ is obtained by sampling an element of $S$ using our $\D$-sampler.
If this sampled element is not in $H$, it is discarded (i.e., $Y_t$ is null); otherwise, $Y_t$ is assigned that element.
We want $\eta$ uniform samples from $H$. 
We can do this by using a simple coupling argument as follows.
For any element $p \in H$, let the probability of it being sampled using our $\D$-sampler be denoted by $\frac{\lambda(p)}{m}$.
Note that $\forall p \in O_{j_i}, \lambda(p) \geq \veps$.
So, for a particular element $p \in H$, $Y_t$ is assigned $p$ with probability $\frac{\lambda(p)}{m} \geq \frac{\veps}{m}$.
One way of sampling a random variable $X_t$ as in Lemma~\ref{lem:clinaba} is as follows --  first sample using $Y_t$. If $Y_t$ is null, then $X_t$ is also null. 
Otherwise, suppose $Y_t$ is assigned an element $p$ of $H$. 
In this case, we define $X_t = p$ with probability $\frac{\veps}{\lambda(p)}$, null otherwise. 
It is easy to check that with probability $\veps$, $X_t$ is a uniform sample from $H$, and null with probability $1-\veps$. 
Now, observe that the set of elements of $H$ sampled by $Y_1, \ldots, Y_N$ is always a superset of $X_1, \ldots, X_N$. 
We can now use Lemma~\ref{lem:clinaba} to finish the proof. \qed
\end{proof}

\end{document}